\documentclass{article}
\usepackage[a4paper,left=3.4cm,right=3.4cm,top=3cm,bottom=4cm,]{geometry}
\usepackage{pdflscape}
\usepackage{amssymb}
\usepackage{amsmath}
\usepackage{amsfonts}
\usepackage{amsthm}
\usepackage{longtable}
\usepackage{array}
\usepackage{varwidth}
\usepackage{braket}
\usepackage{xcolor}
\usepackage{booktabs}
\usepackage{soul}
\usepackage{todonotes}
\usepackage[colorlinks,linkcolor=blue,citecolor=green]{hyperref}
\usepackage[capitalize]{cleveref}
\usepackage{authblk}
\usepackage{cite}
\usepackage{microtype}

\newcommand{\galg}{\mathfrak{g}}
\newcommand{\So}{\mathfrak{so}}
\newcommand{\Sl}{\mathfrak{sl}}

\newcommand{\R}{\mathbb{R}}
\DeclareMathOperator{\ud}{d}

\DeclareMathOperator{\GL}{GL}
\DeclareMathOperator{\tr}{tr}
\DeclareMathOperator{\Span}{span}
\DeclareMathOperator{\diag}{diag}

\renewcommand{\phi}{\varphi}

\newtheorem{theorem}{Theorem}[section]

\newtheorem{corollary}[theorem]{Corollary}
\newtheorem{lemma}[theorem]{Lemma}
\newtheorem{proposition}[theorem]{Proposition}
\theoremstyle{definition}
\newtheorem{definition}[theorem]{Definition}
\newtheorem{example}[theorem]{Example}
\theoremstyle{remark}

\newtheorem{remark}[theorem]{Remark}
\allowdisplaybreaks[3]
\usepackage{graphicx} 
\usepackage{lipsum}

\numberwithin{equation}{section}

\title{Lie algebras with compatible scalar products for non-homogeneous Hamiltonian operators }

\author[1,2]{G. Gubbiotti}
\author[3]{F. Oliveri}
\author[3]{E. Sgroi}
\author[2,4]{P. Vergallo}

{\small \affil[ 1]{\small Universit\`a degli Studi di Milano, Dipartimento di Matematica
``Federigo Enriques'', Via Cesare Saldini 50, 20133, Milano, Italy}
\affil[2]{INFN Sezione di Milano, Via Giovanni Celoria 16, 20133, Milano, Italy}
\affil[3]{Universit\`a degli Studi di Messina, Dipartimento di Scienze Matematiche e Informatiche, Scienze Fisiche e Scienze della Terra, V.le F. Stagno D'Alcontres 31, I-98166 Messina, Italy}
\affil[4]{Universit\`a degli Studi di Messina, Dipartimento di Ingegneria, Contrada Di Dio, 98166 Sant'Agata, Messina, Italy}
}

\date{\today}

\begin{document}
\maketitle

\vspace{10mm}

\begin{abstract}
    We study from an algebraic and geometric viewpoint Hamiltonian operators which are sum of a non-degenerate first-order homogeneous operator and a Poisson tensor. In flat coordinates, also known as Darboux coordinates, these operators are uniquely determined by a triple composed by a Lie algebra, its most general non-degenerate quadratic Casimir and a 2-cocycle. We present some classes of operators associated to Lie algebras with non-degenerate quadratic Casimirs and we give a description of such operators in low dimensions. Finally, motivated by the example of the KdV equation we discuss the conditions of bi-Hamiltonianity of such operators.
\end{abstract}

\newpage

\section{Introduction}

The theory of integrable systems is a broad topic which lies in the exact middle between geometry, differential \cite{DubKri,Olver:ApLGDEq} and algebraic \cite{donsha,pvan} mainly, and mathematical physics, especially in the contexts of mechanics \cite{arnold2007mathematical}, hydrodynamics \cite{NovikovManakovPitaevskiiZakharov:TS} and field theories \cite{dubfield}. As a building block, a key role in this theory is played by Poisson structures and the related Hamiltonian formalism. Such a formalism was developed first for Ordinary Differential Equations (ODEs) and later extended to Partial Differential Equations (PDEs), see \cite{Magri}. In its present form, the Hamiltonain theory of PDEs represents a successful tool to investigate nonlinear phenomena using a geometric approach. We refer to the review paper \cite{mokhov98:_sympl_poiss} for further details. 

Let us specify to the case of evolutionary systems of order $\ell$ in $N$ components,
i.e.\ systems of PDEs of the following form:
\begin{equation}\label{a1}
    u^i_t=F^i(x, u, u_x, \ldots,u_{\ell x}), \qquad i=1,\dots , N,
\end{equation}
where $x,t$ are independent variables, $u^i:X\subset \mathbb{R}^2\longrightarrow \mathbb{R}, i=1,\dots n$, and  $F^i$ are smooth functions in their arguments. In this context, we say that \eqref{a1} is Hamiltonian if there exists a differential operator $\mathcal{A}=(\mathcal{A}^{ij})_{i,j=1,\dots n}$ and a functional $H=\int{h\, dx}$ such that the system can be re-written as
\begin{equation}
    u^i_t=\mathcal{A}^{ij}\left(\frac{\delta H}{\delta u^j}\right), \qquad i=1,\dots, N,
\end{equation}
with $\mathcal{A}$ a Hamiltonian operator (HO) and $\delta/\delta u^j$ the variational derivative with respect to the $j$-th field variable. We recall that the operator $\mathcal{A}$ is said Hamiltonian if for every pair of functionals $F=\int{f\, dx},G=\int{g\, dx}$ the bracket
\begin{equation}\label{bra}
    \{F,G\}_{\mathcal{A}}=\int{\frac{\delta f}{\delta u^i}\mathcal{A}^{ij}\left(\frac{\delta g}{\delta u^j}\right) \, dx}
\end{equation}
define a Poisson bracket, i.e. it is skew-symmetric and satisfies the Jacobi identity. Note that in the finite-dimensional case of ODEs, such a formalism reduces to the well-studied Poisson bracket defined in terms of a Poisson tensor $\omega$ on the field variables manifold.

The geometric nature of the operators is almost clear not only in the framework of classical analytic mechanics, but also when dealing with nonlinear PDEs. As an example, the easiest structure for a purely differential operator is the one introduced by Dubrovin and Novikov in 1983:
\begin{equation}\label{dubnov}
    \mathcal{A}_1^{ij}=g^{ij}\partial_x+b^{ij}_ku^k_x,
\end{equation}
where the coefficients $g$ and $b$ depend on $u^1,\ldots, u^n$ only. Such an operator has a homogeneous degree equal to one according to the natural grading rules $\deg \partial_x^k=k$ and $\deg u_{lx}=l$. In \cite{DN83}, the authors proved that an operator of type \eqref{dubnov} defines a Poisson bracket \eqref{bra} in the non-degenerate case $\det g^{ij}\neq 0$, if and only if  $g_{ij}=(g^{ij})^{-1}$ is a flat metric and the symbols $b^{ij}_k$ satisfy $b^{ij}_k=-g^{is}\Gamma^j_{sk}$, where $\Gamma^i_{jk}$ define the Levi-Civita connection of $g$. In this case, $\mathcal{A}_1$ is also known as a Dubrovin-Novikov operator. 

Since the pioneering work of Dubrovin and Novikov, the theory of homogeneous Hamiltonian operators has been developed from different viewpoints: from the projective geometric one (see \cite{FPV14,VerVit2,LSV,GubvGeVer1}), from the cohomological one (see \cite{PC1,PC2}), from the bi-Hamiltonian perspectives (see \cite{Magri,omokh,LSV,PavVerVit1,LV}), and through the properties of the associated system of quasilinear conservation laws (see \cite{AgaFer,tsa2,tsa1}).

In this paper, we consider non-homogeneous operators of hydrodynamic type, i.e. Hamiltonian operators of type $1+0$:
\begin{equation}\label{op1}
    \mathcal{A}= \mathcal{A}_1\, +\, \mathcal{A}_0,
\end{equation}
where
\begin{align}
    \mathcal{A}_1^{ij}=g^{ij}\partial_x+b^{ij}_ku^k_x\, , \qquad 
\mathcal{A}^{ij}_0=\omega^{ij},
\end{align}
here the coefficients $g^{ij},b^{ij}_k,\omega^{ij}$ depend on the field variables $u^1,\dots, u^n$ only. We also refer to \cite{koinho} for non-homogeneous operators of different type.

Operators of this type arise in different contexts in mathematical physics and geometry. They were firstly introduced by Dubrovin and Novikov in \cite{DubrovinNovikov:PBHT} as a natural extension of first order homogeneous Hamiltonian operators in 1984. Few years later, in 1994,   E.V. Ferapontov and O.I. Mokhov in \cite{MF94} found the necessary and sufficient conditions for them to be Hamiltonian, then  S.P. Tsarev proved that the key integrable model of the KdV equation admits such a Hamiltonian structure after the inversion procedure (see \cite{tsa3}). Moreover, non-homogeneous operators of type $1+0$ (also known as non-homogeneous hydrodynamic operators) are proved to arise as the natural (average) Poisson structure of perturbed systems after the application of Whitham modulation procedure (see \cite{DB2}). More recently, one of us found necessary geometric conditions for a non-homogeneous quasilinear system to admit Hamiltonian formalism with $1+0$ operators \cite{Ver2}. This last result was recently generalized by X. Hu and M. Casati for the multidimensional case \cite{HuCas24}. A recent classification of operators \eqref{op1} was developed by M. Dell'Atti e P. Vergallo for $2$ and $3$ number of components when the leading coefficient $g$ is degenerate (see \cite{DellAVer1}). This last result was recently generalised to the case of $2$ spatial variables (see \cite{Riz1}).

Operators of the form \eqref{op1} are not Hamiltonian in general, indeed the following Theorem holds:
\begin{theorem}[\cite{FerMok1, mokhov98:_sympl_poiss}]\label{thm1}
The operator \eqref{op1} is Hamiltonian if and only if $\mathcal{A}_1$ 
is Hamiltonian, $\mathcal{A}_0$ is Hamiltonian, and the following compatibility conditions are satisfied
\begin{subequations}
\begin{align}\label{cond1}
    \Phi^{i j k} &= \Phi^{k i j} \,,
\\\label{eq:phi}
        \frac{\partial \Phi^{i j k}}{\partial u^r} & = \sum_{(i,j,k)}  b^{s i}_{r}\, \frac{\partial \omega^{j k}}{\partial u^s} + \left( \frac{\partial b^{i j}_{r}}{\partial u^s} - \frac{\partial b^{i j}_{s}}{\partial u^r}  \right)\omega^{s k}\,,
\end{align}
\end{subequations}
where $\Phi^{i j k}$ is the $(3,0)$-tensor
\begin{align}
        \Phi^{i j k} & = g^{i s}\, \frac{\partial \omega^{j k}}{\partial u^s} - b^{i j}_{s} \, \omega^{s k} - b^{i k}_{s} \, \omega^{j s} \,. 
\end{align} 
\end{theorem}

Note that no assumption of non-degeneracy is made up to now, i.e.\ if the conditions of \Cref{thm1} are satisfied the operator is Hamiltonian also when $\det g=0$.

In \cite{DellAVer1}, the authors showed that a large class of scalar evolutionary 
equations $u_t=F(u,\dots , u_{(\ell-1)x})+u_{\ell x}$ can be recast into a non-homogeneous 
quasilinear system of first order PDEs. This procedure is very common for ODEs and was 
deeply investigated by S.P.\ Tsarev with respect to the Hamiltonian property of the 
systems obtained through such transformations \cite{tsa3, Kersten2008HamiltonianSF}. 
In short, being $\ell$ the order of the equation, we introduce new variables $u^1=u,u^2=u_x, \dots , u^\ell=u_{(\ell-1)x}$ so 
that the scalar equation is brought into an evolutionary system in $u^j_x$ of equations 
of first order in $u_t$. Such a procedure is also known as the \emph{inversion} procedure, 
because the resulting system is evolutionary with respect to $x$. 

As an example, let us consider the generalised KdV equation:
\begin{equation}\label{mkdv}
    u_t+3(n+1)\,u^n\,u_x+u_{xxx}=0
\end{equation}
where $n$ is a positive integer. We introduce the variables 
$u^1=u,\,u^2=u_{x},\,u^3=u_{xx}$, and map the equation into 
\begin{equation}\label{sysmkdv}
    \begin{cases}
    u^1_x=u^2\\[1.2ex]
    u^2_x=u^3\\[1.2ex]
    u^3_x=-u^1_t-3(n+1)(u^1)^n \, u^2
    \end{cases} \,.
\end{equation}
It turns out (see \cite{DellAVer1}) that system \eqref{sysmkdv} is Hamiltonian with a non-homogeneous operator \eqref{op1}:  
\begin{equation}\label{opmkdv}
\begin{split}
    A&=\begin{pmatrix}
    0&0&0\\0&0&0\\0&0&1
    \end{pmatrix}\partial_t+ \begin{pmatrix}
    0&1&0\\-1&0&-3(n+1)(u^1)^{n-1}\\
    0&3(n+1)(u^1)^{n-1}&0
    \end{pmatrix},
\end{split} 
\end{equation} 
and the Hamiltonian functional 
\begin{equation}
    H=\int\left( 3(u^1)^{n+1}-u^1u^3+\frac{(u^2)^2}{2}\right) \, dx.
\end{equation}
Note that this class of equations contains the KdV equation itself (for $n=1$) 
and the modified KdV  equation (for $n=2$). 


\vspace{5mm}

\paragraph{Structure of the paper.} The aim of this paper is to give a new insight on non-homogeneous hydrodynamic type operators using the theory of Lie algebras, and retrieving the geometric properties from the Lie algebraic
ones. 
{
The paper is divided into three main sections, plus the conclusions and two appendices. In Section \ref{sec1}, we deepen the study of Hamiltonian structure with non-degenerate leading coefficient, introducing the Darboux form of the operator and discussing the geometric interpretation of non-homogeneous structures in the contest of Lie algebras. In Section \ref{sec3}, we describe classes of Lie algebras giving arise to $1+0$ Hamiltonian operators; we use purely algebraic results to study in details abelian and semi-simple Lie algebras, structures which are direct sum of other ones and the particular case of two-step nilpotent Lie algebras. Furthermore, in Section \ref{sec4}, we present the complete list of Hamiltonian operators up to $n=6$ number of components. Finally, in the Conclusions we draw some further perspectives for pairs of compatible operators of type $1+0$, led by the example of the KdV equation.
}

\section{Non-degenerate structures}\label{sec1}

Let us firstly assume that $g$, viewed as matrix, has non-zero determinant. In this case, 
 we can reformulate Theorem \ref{thm1} as follows
\begin{theorem}
    In the non-degenerate case,  the operator $\mathcal{A}$ is Hamiltonian if and only if $\mathcal{A}_1$ is a Dubrovin-Novikov operator and  $\mathcal{A}_0=\omega$ is a Poisson tensor that is also a Killing-Yano tensor for the metric $g$.
    \label{thm:structgeom}
\end{theorem}
Note that, in coordinates the property of being Killing-Yano {\cite{kY1,kY2}} for $g$ reads as $\nabla^{[i}\omega^{j]k}=0$, or more explicitly 
\begin{equation}\label{form1}
    \nabla^i\omega^{jk}+\nabla^j\omega^{ik}=0,
\end{equation}
where $\nabla^i=g^{is}{\nabla_s}$ and $\nabla_s$ is the covariant derivative with respect to $g_{ij}$. Formula \eqref{form1} corresponds to condition \eqref{cond1} in Theorem \ref{thm1}. We remark that \eqref{eq:phi} is trivially satisfied for every Killing-Yano tensor. 

\vspace{5mm}

\subsection{Darboux form of Hamiltonian operators}
We briefly recall the following definition 
\begin{definition}\label{casy}
    Let $\mathcal{A}$ be a Hamiltonian operator and $\{\,,\,\}_{\mathcal{A}}$ the associated Poisson bracket as in \eqref{bra}. A functional $C=\int{c\, dx}$ is said a Casimir for $\mathcal{A}$ if 
    \begin{equation}
        \{C,F\}_{\mathcal{A}}=0 
    \end{equation}for every other functional $F=\int{f\, dx}$.
\end{definition}
In the particular case of a first-order homogeneous operator \eqref{dubnov}, if $\mathcal{A}_1$ is a non-degenerate operator, then there exist $n$ functionally independent functions $f_l$ such that 
\begin{equation}
    \nabla_i\nabla_j f_l=0 \qquad i,j=1,\dots, n,
\end{equation}which are the Casimir densities of the operator $\mathcal{A}_1$. Introducing the change of variables defined as
\begin{equation*}
    \tilde{u}^1=f_1(u^1,\dots ,u^n)\, ,\quad  \dots , \quad \tilde{u}^n=f_n(u^1,\dots , u^n),
\end{equation*}
the operator simply reduces to $\mathcal{A}^{ij}_1=\eta^{ij}\partial_x$, where $\eta^{ij}\in\mathbb{R}$ (and so all the Christoffel symbols vanish). This form of the operator is also known as \emph{Darboux form} of a first-order Hamiltonian structure, in complete analogy with the Darboux form of a symplectic tensor, see \cite{getz}. So, we reduce Theorem \ref{thm:structgeom} to 
\begin{corollary}[Mokhov, \cite{mokhov98:_sympl_poiss}]\label{corollary1}In Darboux form for $\mathcal{A}_1^{ij}=\eta^{ij}\partial_x$, the operator $\mathcal{A}^{ij}$ is Hamiltonian if and only if $\mathcal{A}_0^{ij}$ is linear in $u^k$
\begin{equation}
  \mathcal{A}_0^{ij}=  \omega^{ij}=c^{ij}_ku^k+f^{ij},
\end{equation}
such that $c^{ij}_k$ are structure constants of a real Lie algebra $\mathfrak{g}$, $f$ is a 2-cocycle on it and $\eta$ is a scalar product compatible with $\mathfrak{g}$.
    \label{cor:mokhov10}
\end{corollary}
This version of the Hamiltonian conditions is the starting point of our investigation, so it is worth to recall the main algebraic structures here used. In particular, given a Lie algebra $\mathfrak{g}$ we say that $f$ is a 2-cocycle if $f\in \Lambda^2\mathfrak{g}^*$  such that its \emph{coboundary} $\delta f$ annihilates, where 
\begin{equation*}
    \delta f: \mathfrak{g}\otimes \mathfrak{g}\otimes \mathfrak{g}\longrightarrow \mathbb{R}, \qquad  \delta f(x,y,z)=f([x,y],z)+<cyc>.
\end{equation*}
 In coordinates, we read the present condition as
\begin{equation}
    \label{compa1}c^{ij}_sf^{sk}+c^{jk}_sf^{si}+c^{ki}_sf^{sj}=0.
\end{equation}
Moreover, we say that a scalar product $\langle\, ,\, \rangle$ is compatible with a Lie algebra $\mathfrak{g}$ if for all $x,y,z\in\mathfrak{g}:$
$   \langle \text{ad}_xy,z\rangle=\langle x,\text{ad}_zy\rangle$,
whose coordinate description in our context is
\begin{equation}\label{compa2}
    \eta^{is}c^{jk}_s+\eta^{js}c^{ik}_s=0.
\end{equation}

Analogously, let us recall that given  
$\Lambda^1 \galg^* = \Span\Set{\theta_1,\dots,\theta_{n}}$, then:
\begin{equation}
    \ud \theta_k = -\frac{1}{2} c^{ij}_{k}\theta_{i}\wedge\theta_{j}.
    \label{eq:extdev}
\end{equation}
So that, we can use \eqref{eq:extdev} to express the 2-cocycle condition in coordinate-free form. Indeed, it is easy to show that given $f\in\Lambda^2\mathfrak{g}^*$, $f$ is a 2-cocycle if and only if $\ud f=0$. To this aim, let us consider a 2-form $f=f^{ij}\theta_i\wedge \theta_j$, so that 
\begin{align}\begin{split}
    \ud f&= f^{ij}\, \ud \theta_i \wedge \theta_j - f^{ij}\, \theta_i\wedge \ud \theta_j\\
    &=-\frac{1}{2}f^{ij}c_{i}^{as}\, \theta_a\wedge \theta_s \wedge \theta_j +\frac{1}{2} f^{ij}c_{j}^{as} \, \theta_i\wedge \theta_a\wedge \theta_s\\
    &=f^{ij}c_{j}^{as}\, \theta_i\wedge \theta_a\wedge \theta_s.\end{split}
\end{align}

Finally, the previous expression reads as \eqref{compa1}.

\begin{remark}
    We remark that in the majority of references 
    for Lie algebras, see e.g.  {\cite{Jacobson1962,Kirillov2008}},
    the structure constants on $\mathfrak{g}$ are taken as tensors of type $(1,2)$, 
    i.e. they are indicated by $c^i_{jk}$. In the previous formulas we 
    use the contravariant description of Poisson bivectors $\mathcal{P}^{ij}$ with 
    upper indices. However, since all results on Lie algebras are 
    invariant with respect to linear transformations, the choice of upper or
    lower indices is inessential.
    \label{rem:contravariant}
\end{remark}

Here we present an example of system coming from the theory of interacting waves which admits a Hamiltonian structure of the described type.  

\begin{example}\label{3waves}A physical example of system admitting a Hamiltonian structure with non-homogeneous hydrodynamic type operators is given by the 3-waves equations \cite{mokhov98:_sympl_poiss}:
\begin{equation}\begin{cases}\label{3wav}
u^1_t=-c_1u^1_x-2(c_2-c_3)u^2u^3\\[1.2ex]
u^2_t=-c_2u^2_x-2(c_1-c_3)u^1u^3\\[1.2ex]
u^3_t=-c_3u^3_x-2(c_2-c_1)u^1u^2\end{cases} \,,
\end{equation}
that is Hamiltonian with the operator
\begin{equation}C^{\,ij}=
\begin{pmatrix}
1&0&0\\0&-1&0\\0&0&-1
\end{pmatrix}\partial_x+\begin{pmatrix}
0&-2u^3&2u^2\\
2u^3&0&2u^1\\-2u^2&-2u^1&0
\end{pmatrix} \,.
\end{equation}
and the Hamiltonian density $h(u^1,u^2,u^3)=(u^1)^2-(u^2)^2-(u^3)^2$.
\end{example}

\vspace{3mm}

Before investigating the algebraic interpretation of non-homogeneous operators in Darboux form, we show that they are invariant under linear transformations of the field variables $u^1,\dots u^n$. Such type of transformations are the natural ones in the theory of Lie algebras, so we are allowed to consider theoretical results coming from Lie algebras and extend them to the context of non-homogeneous Hamiltonian structures and vice versa. We start from the following very well-known Lemma:

\begin{lemma}\label{lem:changebasis}
The structure constants of an $n$-dimensional Lie algebra $\mathfrak{g}$ are tensors invariant under invertible linear transformations. That is given $A=(a^i_j)\in\GL(n)$ and $B=(b^i_j)$ its inverse, the 
 transformation $ \tilde{u}^i=T^i(u)=a^i_lu^l$ acts on the structure constants in the following way
\begin{equation}
    \tilde{c}^{ij}_k=a^i_la^j_mc^{lm}_sb^s_k.
\end{equation}  
In addition, the constants $\tilde{c}^{ij}_k$ define uniquely a Lie algebra structure isomorphic to $\mathfrak{g}$.
\end{lemma}

So that the invariance can be explicitly proved:

\begin{theorem}\label{thm:trans}
    Invertible linear transformations of the field variables preserve the Darboux form of a non-homogeneous operator.
\end{theorem}
\begin{proof}
    Let us first recall that $\eta^{ij}$ and $\omega^{ij}$ transform as $(2,0)$ tensors, i.e. they satisfy
    \begin{equation}\label{trasf}
        \tilde{\eta}^{ij}=\frac{\partial \tilde{u}^i}{\partial u^k}\eta^{kl}\frac{\partial \tilde{u}^j}{\partial u^l}, \qquad \tilde{\omega}^{ij}=\frac{\partial \tilde{u}^i}{\partial u^k}\omega^{kl}\frac{\partial \tilde{u}^j}{\partial u^l}.
    \end{equation}
    Then, let us consider an invertible linear transformation of the field variables
    \begin{equation}
        \tilde{T}: u^i\mapsto \tilde{u}^i=a^i_lu^l,
    \end{equation} where $a^i_j$ are constants. Now, let us note that $\eta^{ij}\in\mathbb{R}$ are symmetric and $\omega^{ij}=c^{ij}_ku^k+f^{ij}$ as in Corollary \ref{corollary1}. By applying the transformation rule \eqref{trasf}, we have
\begin{align*}
&\tilde{\eta}^{ij} = a_k^i \eta^{kl}a_l^j,
\\
&\tilde{\omega}^{ij} = a_k^i(c_s^{kl}u^s + f^{kl})a_l^j = a_k^i c_s^{kl} a_l^j u^s + a_k^i f^{kl}a_l^j = a_k^i c_s^{kl} a_l^j b_r^s u^r + a_k^i f^{kl}a_l^j,
\end{align*}
where $b_r^s$ are the components of the inverse transformation matrix. Notice that, due to Lemma \ref{lem:changebasis},  $\tilde{c}_r^{ij} = a_k^i a_l^j b_r^s c_s^{kl}$  represent the structure constants of a Lie algebra isomorphic to the original one. Moreover, we can see that $\tilde{\eta}$ is a scalar product (symmetry follows directly) which is compatible with the Lie algebra defined by $\tilde{c}_k^{ij}$ 
\begin{align*}
\tilde{\eta}^{is} \tilde{c}_s^{jk} + \tilde{\eta}^{js} \tilde{c}_s^{ik} & = (a_\alpha^i \eta^{\alpha\beta}a_\beta^s)(a_{\alpha^\prime}^j a_{\beta^\prime}^k c_\gamma^{\alpha^\prime \beta^\prime} b_s^\gamma) + (a_\varphi^j \eta^{\varphi\theta}a_\theta^s)(a_{\varphi^\prime}^i a_{\theta^\prime}^k c_{\gamma^\prime}^{\varphi^\prime \theta^\prime}b_s^{\gamma^\prime}) \\
& = a_\alpha^i a_{\alpha^\prime}^j a_{\beta^\prime}^k \eta^{\alpha\beta} c_\beta^{\alpha^\prime \beta^\prime} + a_{\varphi^\prime}^i a_\varphi^j a_{\theta^\prime}^k \eta^{\varphi \theta} c_\theta^{\varphi^\prime \theta^\prime}  = 0,
\end{align*}
where in the last step we used condition \eqref{compa2}. Finally, setting $\tilde{f}^{ij} = a_k^if^{kl}a_l^j$, it is easy to verify that $\tilde{f}^{ji} = -\tilde{f}^{ij}$, and that
\begin{align*}
\tilde{c}_s^{ij} \tilde{f}^{sk} + \tilde{c}_s^{jk} \tilde{f}^{si} + \tilde{c}_s^{ki} \tilde{f}^{sj} & = a_\alpha^i a_\beta^j a_\gamma^k c_r^{\alpha\beta} f^{r\gamma} + a_\alpha^j a_\beta^k a_\gamma^i c_r^{\alpha\beta} f^{r\gamma} + a_\alpha^k a_\beta^i a_\gamma^j c_r^{\alpha\beta} f^{r \gamma} \\
& = -a_\alpha^i a_\beta^j a_\gamma^k c_r^{\beta\gamma}f^{r\alpha} - a_\alpha^i a_\beta^j a_\gamma^k c_r^{\gamma \alpha} f^{r\beta}  \\
&\hphantom{ciaociao}+ a_\alpha^j a_\beta^k a_\gamma^ic_r^{\alpha\beta}f^{r\gamma} + a_\alpha^k a_\beta^i a_\gamma^j  c_r^{\alpha\beta} f^{r\gamma}  \\
&= -a_\alpha^i a_\beta^j a_\gamma^k c_r^{\beta\gamma}f^{r\alpha} - a_\alpha^i a_\beta^j a_\gamma^k c_r^{\gamma \alpha} f^{r\beta}  \\
&\hphantom{ciaociao}+ a_\alpha^i a_\beta^j a_\gamma^k c_r^{\beta\gamma}f^{r\alpha} + a_\alpha^i a_\beta^j a_\gamma^k c_r^{\gamma \alpha} f^{r\beta}   = 0,
\end{align*}
where in the last step we used condition \eqref{compa1} and we re-labeled the indices.
\end{proof}

\vspace{3mm}


The present results allow to interpret in purely algebraic terms the Hamiltonian operators on which we are focusing.

\subsection{Algebraic interpretation of 1+0 operators in Darboux form}

In accordance with the previous paragraph, we provide a purely algebraic reformulation of the \Cref{cor:mokhov10}.
To do so, we recall some known definitions from the literature. For references, we address to the well-known textbooks 
\cite{Jacobson1962,Kirillov2008,FultonHarris1991} and the comprehensive monograph 
{\cite{snowin}}. 

Let us consider an $n$-dimensional Lie algebra 
$\galg$ over a field $\mathbb{K}$. Associated to $\galg$, 
there is the associative algebra $U\galg$ which can be constructed as follows:
\begin{equation}
    U\galg = \bigoplus_{k=0}^{\infty} \galg^{\otimes k}/J,
    \quad\quad 
    \galg^0 = \mathbb{K},
\end{equation}
where $J$ denotes the ideal:
\begin{equation}
    J= I(x y - y x - [x,y])_{x,y\in\galg}.
\end{equation}
In $U\galg$ the Lie bracket is given, as in any associative algebra, by
the usual commutator $[x,y]_{U\galg} = x y - y x$. This structure is also known as \emph{universal enveloping algebra} (UEA).

\begin{definition}\label{defcasel}
    An element $C\in U\galg$ is said to be a \emph{Casimir element} if is
    invariant with respect to the adjoint action of $\galg$, defined by extending 
    the adjoint action of $\galg$ on $U\mathfrak{g}$.
\end{definition}

Then, Casimir elements have the following characterisation:

\begin{proposition}[\cite{Kirillov2008}]
    Given a Lie algebra $\galg=\Span\Set{ e^1,\ldots,e^n}$, the space of Casimir elements of $U\galg$ coincides with the
    centre $Z(U\galg)$, i.e. the set of elements $C\in U\galg$
    commuting with \emph{all other elements}:
    \begin{equation}
    [C,e^i]_{U\galg} = Ce^i - e^i C =0, \quad
    i = 1,\ldots, n.
    \label{eq:casdefei}
\end{equation}
\end{proposition}


In what follows a particular r\^ole will be played by the so-called quadratic 
Casimir polynomials whose coordinate expression is the following
\begin{equation}\label{casimir}
    C = \frac{a_{ij}}{2}\left(e^i e^j+e^j e^i\right),
\end{equation}
where $a_{ij}\in\mathbb{K}$. Additionally, a quadratic Casimir
polynomial $C$ is said to be non-degenerate if the associated symmetric 
matrix $(a_{ij})_{i,j=1,\ldots,n}\in M_{n,n}(\mathbb{K})$ is non-degenerate, 
i.e. $\det (a_{ij})\neq 0$. We then recall the following proposition which shows how to explicitly 
compute the quadratic Casimir polynomials of a given Lie algebra $\galg$:

\begin{proposition}[\cite{Kirillov2008}]\label{prop:casimircalc}
    Let $\galg$ be an $n$-dimensional real Lie algebra
    with structure constants $c^{ij}_{k}$. Then, an element $C\in U\galg$ is
    a quadratic Casimir polynomial if and only if it satisfy the following 
    algebraic equations:
    \begin{equation}
        a_{is}c^{sk}_j + a_{js}c^{sk}_{i} = 0.
        \label{eq:quadcascond}
    \end{equation}
\end{proposition}

The proof of this fact follows from using the definition of Casimir and
the fact that on $U\galg$ we have 
   $ e_ie_j - e_je_i - c_{ij}^{k} e_k \equiv 0.$
   
Using the previous preliminary definitions and results, we obtain the main theorem of this section: this gives a concrete interpretation of $1+0$ HOs in the context of Lie algebras.

\begin{theorem}\label{thm:principale}
    Non-degenerate quadratic Casimir invariants of a Lie algebra are in bijective correspondence with scalar products which are compatible with the same algebra. In particular, the matrix associated to the bilinear form given by the Casimir is the inverse of the one associated to the scalar product.
\end{theorem}

\begin{proof}
    Let $C$ be a quadratic and non-degenerate Casimir for $\mathfrak{g}$.
    Then, from Proposition \ref{prop:casimircalc} its coefficients satisfy
    the system \eqref{eq:quadcascond}. By the non-degeneracy assumption, 
    let $\eta^{ij}$ be the inverse  matrix of $a_{ij}$, i.e. $\eta^{is}a_{sj}=\delta_j^i$.
    Contracting \eqref{eq:quadcascond} with $\eta^{li}\eta^{mj}$ we 
    obtain:
    \begin{equation}
        \eta^{mj}c_j^{lk}+\eta^{li}c_i^{mk}=0,
        \label{eq:contra}
    \end{equation}
    i.e. 
    it is given by \eqref{compa2}. So $\eta$ defines a scalar product compatible 
    with the Lie algebra $\mathfrak{g}$. Noting that we can reverse all the steps in the previous reasoning, the Theorem is proved.
\end{proof}
As a corollary, we obtain the following:
\begin{corollary}The number of free parameters in the general form of $\eta$ for fixed Lie algebra $\mathfrak{g}$ coincides with the dimension of the space of non-degenerate quadratic Casimirs of $\mathfrak{g}$.
\end{corollary}

\subsection{Final remarks on Casimir functionals of 1 + 0 HOs}

A clarification is finally needed on the distinction between Casimir functionals and the Casimir elements of a Lie algebra, as respectively introduced in definitions \ref{casy} and \ref{defcasel}. In order to do that, we remark that they coincide only for linear Poisson brackets, i.e. Poisson brackets defined by a Lie algebra (see \cite[Chapter 4]{PC1}). In those cases, we have 
\begin{equation}
    \{f,g\}_\omega=\frac{\partial f}{\partial u^i}\omega^{ij}\frac{\partial g}{\partial u^j}=\frac{\partial f}{\partial u^i}c^{ij}_su^s\frac{\partial g}{\partial u^j}
\end{equation}
so that a Casimir function $C(u)$ satisfies 
 \begin{equation}\label{34567}
        c^{ij}_ku^k\frac{\partial C}{\partial u^j}=0.
    \end{equation}
It is remarkable that expression \eqref{34567} is equivalent to  \eqref{eq:casdefei} in the coadjoint representation of a Lie algebra $\galg$ in the space of vector fields of a manifold. This means that for Poisson tensors $\omega^{ij}=c^{ij}_ku^k$ (here $f^{ij}=0$ identically), the Casimirs of the algebra associated to $c^{ij}_k$ coincide with the Casimir functions of the ultralocal operator. 

In spite of this result, a similar statement does not hold for the whole $1+0$ HO. However, we can prove the following theorem: 
{
\begin{theorem}\label{gencasthm}
    In the non-degenerate case, operators of type $
1+0$ in Darboux form admit only linear Casimirs. Moreover, the Casimir functions $C(u)$ are linear combination of elements in the center $Z(\mathfrak{g})$ of the algebra satisfying $F\cdot \nabla C=0$, where $F$ is the matrix associated to the 2-cocycle. 
\end{theorem}
\begin{proof}
    Let $\mathcal{A}_1^{ij}=\eta^{ij}\partial_x$ be a Hamiltonian operator in Darboux form, i.e. $\eta$ is a symmetric matrix. Then, $C(u)$ is a Casimir for $\mathcal{A}_1$ if and only if
    \begin{equation}
        0=\mathcal{A}^{ij}_1\frac{\partial C}{\partial u^j}=\eta^{ij}\partial_x\frac{\partial C}{\partial u^j}=\eta^{ij}\frac{\partial^2 C}{\partial u^j\partial u^s}u^s_x,
    \end{equation}
    that is $C(u)=a_iu^i$, where $a_i$ are constants.

    We now remark that a hydrodynamic density $C(u)$ is a Casimir for the non-homogeneous operator $\mathcal{A}$ if and only if it is a Casimir of $\mathcal{A}_1$ and $\mathcal{A}_0$ simultaneously. Finally, we recall that a linear function is a Casimir of a linear Poisson tensor if 
    \begin{equation}
        0=\omega^{ij}\frac{\partial C}{\partial u^j}=\left(c^{ij}_ku^k+f^{ij}\right)a_j.
    \end{equation}
    Being the previous expression a polynomial in $u^\ell$, we obtain that $C$ must be a Casimir for the Lie algebra associated to $c^{ij}_k$ and must satisfy the additional hypothesis of the present Theorem. 
\end{proof}
\begin{remark}
If $C(u)=a_iu^i+a_0$, with $a_i\neq 0$ for every $i=1,2,\dots n$, then $u^i\in Z(\mathfrak{g})$, that is the Lie algebra is abelian. Then, $c^{ij}_k=0$ with the additional requirement that the following linear system must be satisfied in $a_i:$ $$f^{ij}a_j=0, \qquad i=1,2,\dots n.$$
In particular, the solution depends on $\text{rank}(f^{ij})$: in the most general case the rank is maximal, i.e. $\text{rank}(f^{ij})=n$ if $n$ is even and $\text{rank}(f^{ij})=n-1$ whenever $n$ is odd. Furthermore, in the even case (with general 2-cocycles) the only Casimir functions are constants.
\end{remark}

As a consequence of Theorem \ref{gencasthm} we obtain that 
\begin{corollary}\label{thmcas}
    In the non-degenerate case, operators of type 1+0 have no non-degenerate quadratic Casimir functions. 
\end{corollary}}



{From this point forward, we will not refer anymore to the Casimirs functionals of the operator. Therefore, whenever we will use this term we will refer to the Casimir elements of the Lie algebra associated to the structure constants $c^{ij}_k$.}

\section{Some relevant classes of Lie algebras and associated operators}\label{sec3}

In this Section we describe some classes of Lie algebras for which
we can prove, in full generality, that they admit a non-degenerate
quadratic Casimir polynomial. Based on \Cref{thm:principale} those classes will be associated
to a 1+0 non-homogeneous Hamiltonian operator. 
\subsection{Abelian Lie algebras}
\label{sss:abelian}

Let us begin with the trivial case of abelian Lie algebras. The following result holds true:

\begin{theorem}
    Let $n\mathfrak{n}_{1,1}=\Span\Set{ e^1,\ldots,e^n}$ be the $n$ dimensional 
    abelian Lie algebra.
    Then any quadratic form:
    \begin{equation}
        C = a_{ij}e^ie^j
        \label{eq:na1gen}
    \end{equation}
    is a quadratic Casimir.
    \label{thm:abelian}
\end{theorem}

This result follows from noticing that for an abelian Lie algebra
all generators $e^i$ are \emph{linear} Casimir polynomials, so it is possible
to build a quadratic Casimir polynomial by taking a generic quadratic
combination. This leads to the following 1+0 non-homogeneous Hamiltonian operator of hydrodynamic type:
\begin{equation}
    \label{eq:abeliannhho}
    \mathcal{A}=
    \begin{pmatrix}
        a^{11} & a^{12} & \ldots & a^{1n}
        \\
        a^{12} & a^{22} & \ldots & a^{2n}
        \\
        \vdots & & \ddots & \vdots
        \\
        a^{1n} & a^{2n} & \ldots & a^{nn}
    \end{pmatrix}
    \partial_x +
    \begin{pmatrix}
        0 & f^{1 2} & \ldots & f^{1 n}
        \\
        -f^{1 2} & 0 & \ldots & f^{2n}
        \\
        \vdots & & \ddots & \vdots
        \\
        -f^{1n} & f^{2 n} & \ldots & 0
    \end{pmatrix},
\end{equation}
where $(a^{ij})$ is the inverse matrix of $(a_{ij})$ and $f^{ij}$ are arbitrary constants. Indeed, a completely generic skew-symmetric matrix
is a 2-cocycle for $n\mathfrak{n}_{1,1}$. Since $\det (a^{i j})\neq 0$ for
generic values of the coefficients $a^{ij}$ this is clearly a non-degenerate
operator.


\begin{remark}
    Let us observe that the operator \eqref{eq:abeliannhho} is
    constant, so it is essentially trivial. The reason is that the
    field variables $u^i$ can appear only on non-zero entries in the
    commutation table of the Lie algebra.
\end{remark}

\subsection{Semi-simple Lie algebras}\label{semisimple-sub}

The second easiest case is the one of semi-simple Lie algebras for which the following result holds true:

\begin{theorem}
    Let $\galg$ be a semi-simple Lie algebra. Then, $\galg$ admits
    a non-degenerate quadratic Casimir whose associated matrix is a scalar
    multiple of the \emph{inverse} of the matrix associated to the Killing form.
\end{theorem}

This result follows easily from the fact that the only quadratic symmetric 
invariant with respect to the adjoint action of $\galg$ is the Killing form (see 
\cite[Theorem 5.53]{Kirillov2008}). 

\begin{remark}
    We recall that given any Lie algebra $\galg$ with structure
    constants $c^{ij}_k$ then the explicit expression of the Killing
    form is the following:
    \begin{equation}
        K^{ij} = c^{il}_{m}c_{l}^{jm},
        \label{eq:Kijcoord}
    \end{equation}
    while for matrix Lie groups it is related to the trace.
   \end{remark}

\begin{remark}
    We additionally remark that semi-simple Lie algebras have vanishing
    second cohomology group. This implies that to construct
    the associated 2-cocycle we can simply consider a generic 1-form
    on the algebra, compute its exterior derivative using formula
    \eqref{eq:extdev}, and then take its coefficients.
    \label{rem:ss2cocycles}
\end{remark}

We discuss now some relevant examples of
1+0 HOs obtained from such a class of Lie algebras.

\paragraph{The $\Sl(2,\R)$ algebra.}
    \label{ex:sl2}
    Consider the $\Sl(2,\R)=\Span\Set{ J_+,J_-,J_3}$ Lie algebra 
    with commutation relations:
    \begin{equation}
        [J_-,J_+] = 4J_3,
        \quad
        [J_3,J_+] = 2J_+,
        \quad
        [J_3,J_-] =- 2J_-.
    \end{equation}
    It is well known that the associated Killing form is:
    \begin{equation}
        K(\Sl(2,\R)) =
        \begin{pmatrix}
            0 & -16 & 0
            \\
            -16 & 0 & 0
            \\
            0 & 0 & 8
        \end{pmatrix},
    \end{equation}
    so that we can take as a quadratic Casimir the polynomial:
    \begin{equation}
        C = J_+J_- + J_- J_+   - 2J_3^2.
        \label{eq:carsl2}
    \end{equation}
    By a direct computation we see that also in this case we can take as 
    2-cocycle any skew-symmetric matrix, such that following the procedure 
    highlighted above we can build the 1+0 HO:
    \begin{equation}
        \mathcal{A}=
        a
        \begin{pmatrix}
            0 & 2 & 0 \\
            2 & 0 & 0 \\
            0 & 0 & -1 \\
        \end{pmatrix}
        \partial_x +
        \begin{pmatrix}
            0 & -4u^3 & -2u^+ \\
            4u^3 & 0 & 2u^- \\
            2u^+ & -2u^- & 0
        \end{pmatrix}
        +
        \begin{pmatrix}
            0 & f^{12} & f^{13} \\
            -f^{12} & 0 & f^{23} \\
            -f^{13} & -f^{23} & 0
        \end{pmatrix},
    \end{equation}
    where $\Set{u^{+},u^{-},u^3}$ is a basis for the
    (commuting) variables in the symmetric algebra, and $a\in\R$ 
    is an arbitrary constant.

Using the isomorphism $\mathfrak{sl}(2,\mathbb{R})\cong \mathfrak{so}(1,2)$ we can also recast the above operator into the form 
\begin{equation}
    \mathcal{A}'=a\begin{pmatrix}
1&0&0\\0&-1&0\\0&0&-1
\end{pmatrix}\partial_x+\begin{pmatrix}
0&-2u^3&2u^2\\
2u^3&0&2u^1\\-2u^2&-2u^1&0
\end{pmatrix}+ \begin{pmatrix}
            0 & f^{12} & f^{13} \\
            -f^{12} & 0 & f^{23} \\
            -f^{13} & -f^{23} & 0
        \end{pmatrix}.
\end{equation}

We notice that choosing $a=1$ and $f^{12}=f^{13}=f^{23}=0$, we obtain the Hamiltonian operator for the 3-wave interacting system in Example \ref{3waves}. The Hamiltonian operator here considered is also used to prove the bi-Hamiltonian property of the KdV equation after the inversion procedure in the local quadratic unimodular change of variables (see \cite{mokhov98:_sympl_poiss} and subsection \ref{kdvsection} of the present paper).

\begin{remark}
    This case can be rephrased also with the well-known matrix
    representation of $\Sl(2,\R)$ in terms of traceless matrices in dimension 
    two, see~\cite[Section 10.4]{FultonHarris1991}:
    \begin{equation}
        H = 
        \begin{pmatrix}
            1 & 0
            \\
            0 & -1
        \end{pmatrix},
        \quad
        X = 
        \begin{pmatrix}
            0 & 1
            \\
            0 & 0
        \end{pmatrix},
        \quad
        Y = 
        \begin{pmatrix}
            0 & 0
            \\
            1 & 0
        \end{pmatrix},
    \end{equation}
    with commutation relations:
    \begin{equation}
        [H,X]=2X,
        \quad
        [H,Y]=-2Y,
        \quad
        [X,Y]=H.
    \end{equation}
    The Killing form can be obtained from direct computation 
    from equation~\eqref{eq:Kijcoord}. 
    
    In the next paragraph, we 
    will generalise this approach.
    \label{rem:sl2mat}
\end{remark}

\paragraph{The $\Sl(n+1,\R)$ algebra.}
    \label{ex:sln}
    An immediate generalisation of the previous remark is
    the special linear Lie algebra in $n+1$ dimensions, $\Sl(n+1,\R)$
    for $n\geq 1$.
    We consider this algebra as realised in terms of traceless matrices:
    \begin{equation}
        \Sl(n+1,\R) = 
        \Set{M \in M_{n+1,n+1}(\R) | \tr M = 0},
    \end{equation}
    and it has dimension $n(n+2)$ (see \cite[Section 2.7]{Kirillov2008}).
    A possible parametrisation of the elements
    of $\Sl(n+1,\R)$ is the following:
    \begin{equation}
        M =
        \begin{pmatrix}
            \omega_1 & m_{1\,2} & \ldots & m_{1\,n+1}
            \\
            m_{2\,1} & \omega_2 & \ldots & m_{2\,n+1}
            \\
            \vdots & & \ddots & \vdots
            \\
            m_{n+1\,1} & m_{n+1\,2} & \ldots & \omega_{n+1}
        \end{pmatrix},
        \quad
        \omega_1 + \ldots + \omega_{n+1} = 0,
    \end{equation}
    or introducing the matrices:
    \begin{equation}
        E_{i,j} = (\delta_{ik}\delta_{jl})_{k,l=1}^{n},
        \quad
        H_i = E_{i,i}-E_{n+1,n+1},
    \end{equation}
    we have:
    \begin{equation}
        M = \sum_{i=1}^n \omega_i H_i
        +
        \sum_{i\neq j}^n m_{i\,j}E_{i,j}.
    \end{equation}
    The elements $H_i$ form the \emph{Cartan subalgebra} of 
    $\Sl(n+1,\R)$, i.e. a maximal abelian subalgebra consisting 
    of elements whose adjoint endomorphism is diagonalizable. From the 
    well-known relation:
    \begin{equation}
        E_{i,j} \cdot E_{k,m} = \delta_{jk} E_{i,m},
        \label{eq:EijEkm}
    \end{equation}
    where $\delta_{j,k}$ is the Kroenecker delta, see e.g.~\cite[Section IV.6]{Jacobson1962} we derive the
    general commutation relations of $\Sl(n+1,\R)$:
    \begin{subequations}
        \begin{align}
        [H_i,H_j] &= 0,
        \\
        [H_i,E_{k,m}] &= \delta_{i,k} E_{i,m} - \delta_{i,m} E_{k,i}
        -\delta_{l+1,k} E_{l+1,m} + \delta_{l+1,m} E_{k,l+1},
        \\
        [E_{i,j},E_{k,m}] &= 
        \delta_{j,k} E_{i,m} - \delta_{i,m} E_{k,j}.
        \end{align}
    \end{subequations}

    Moreover, it is know that (see for instance
    \cite[Exercise 14.36]{FultonHarris1991}) the Killing form
    is a multiple of $\tr(M\cdot M^\prime)$, for 
    $M,M^\prime\in\Sl(n+1,\R)$. So, up to a multiple we have:
    \begin{equation}
        K(\Sl(n+1,\R))(M,M^{\prime})
        =\sum_{i=1}^{n}m_{i\,i}m_{i\,i}^{\prime}
        -\sum_{i=1}^{n}m_{i\,i}\sum_{i=1}^{n}m_{i\,i}^{\prime}
        +\sum_{i\neq j}^{n+1} m_{i\,j}m_{j\,i}^{\prime}.
    \end{equation}
    The associated matrix can be computed through as
    the Jacobian matrix of $K(\Sl(n+1,\R))(M,M^{\prime})$
    with respect to the entries of $M$ and $M^{\prime}$.
    {
    Furthermore, using Remark \ref{rem:ss2cocycles}, we can explicitly compute the generic 2-cocycle of $\Sl(n+1,\R)$ from the generic 1-form on $\Sl(n+1,\R)$.}
    
    For instance, taking $n=3$, we get the following 1+0 Hamiltonian
    operator in dimension 8 with field variables $\Set{u^{i,i}}_{i=1}^{3} 
    \cup \Set{u^{i,j}}_{i\neq j}^{4}$:
    \begin{equation}
        \mathcal{A}_{\Sl(3,\R)}
        = K_3 \partial_x + \Omega_3 + F_3,
    \end{equation}
    where
    \begin{equation}
        K_3 =\left(
\begin{array}{cccccccc}
 2 \alpha & \alpha & 0 & 0 & 0 & 0 & 0 & 0 \\
 \alpha & 2 \alpha & 0 & 0 & 0 & 0 & 0 & 0 \\
 0 & 0 & 0 & 0 & \alpha & 0 & 0 & 0 \\
 0 & 0 & 0 & 0 & 0 & 0 & \alpha & 0 \\
 0 & 0 & \alpha & 0 & 0 & 0 & 0 & 0 \\
 0 & 0 & 0 & 0 & 0 & 0 & 0 & \alpha \\
 0 & 0 & 0 & \alpha & 0 & 0 & 0 & 0 \\
 0 & 0 & 0 & 0 & 0 & \alpha & 0 & 0 \\
\end{array}
\right)      
    \end{equation}
    is the Killing form in matrix form (here $\alpha$ is an arbitrary constant)
    \begin{equation}
        \Omega_3
        =\left(
\begin{array}{cccccccc}
 0 & 0 & u^3 & 2 u^4 & -u^5 & u^6 & -2 u^7 & -u^8 \\
 0 & 0 & -u^3 & u^4 & u^5 & 2 u^6 & -u^7 & -2 u^8 \\
 -u^3 & u^3 & 0 & 0 & u^1-u^2 & u^4 & -u^8 & 0 \\
 -2 u^4 & -u^4 & 0 & 0 & -u^6 & 0 & u^1 & u^3 \\
 u^5 & -u^5 & u^2-u^1 & u^6 & 0 & 0 & 0 & -u^7 \\
 -u^6 & -2 u^6 & -u^4 & 0 & 0 & 0 & u^5 & u^2 \\
 2 u^7 & u^7 & u^8 & -u^1 & 0 & -u^5 & 0 & 0 \\
 u^8 & 2 u^8 & 0 & -u^3 & u^7 & -u^2 & 0 & 0 \\
\end{array}
\right) 
    \end{equation}
    is the matrix of the commutation constants,
    $F_3$ is the 2-cocycle matrix:


\begin{equation}
F_3=\left(
\begin{array}{cccccccc}
 0 & 0 & -f^{23} & 2 f^{24} & -f^{25} & f^{16} & 2 f^{27} & f^{18} \\
 0 & 0 & f^{23} & f^{24} & f^{25} & 2 f^{16} & f^{27} & 2 f^{18} \\
 f^{23} & -f^{23} & 0 & 0 & f^{35} & f^{24} & f^{18} & 0 \\
 -2 f^{24} & -f^{24} & 0 & 0 & -f^{16} & 0 & f^{47} & -f^{23} \\
 f^{25} & -f^{25} & -f^{35} & f^{16} & 0 & 0 & 0 & f^{27} \\
 -f^{16} & -2 f^{16} & -f^{24} & 0 & 0 & 0 & f^{25} & f^{47}-f^{35} \\
 -2 f^{27} & -f^{27} & -f^{18} & -f^{47} & 0 & -f^{25} & 0 & 0 \\
 -f^{18} & -2 f^{18} & 0 & f^{23} & -f^{27} & f^{35}-f^{47} & 0 & 0 \\
\end{array}
\right).
\end{equation}

\paragraph{The $\mathfrak{so}(3,\R)$ algebra.}
    \label{ex:so3}
    Consider the Lie algebra 
    \begin{equation}
        \mathfrak{so}(3,\R)=\Span\Set{ L_1,L_2,L_3}
    \end{equation}
    with commutation relations:
    \begin{equation}
        [L_1,L_2] = L_3,
        \quad
        [L_2,L_3] = L_1,
        \quad
        [L_3,L_1] = L_2.
    \end{equation}
    It is well known that the associated Killing form is:
    \begin{equation}
        K(\mathfrak{so}(3,\R)) = \diag (-2,-2,-2),
    \end{equation}
    so that we can take as a quadratic Casimir the polynomial:
    \begin{equation}
        C = \frac{1}{2}\left(L_1^2+L_2^2+L_3^2\right).
    \end{equation}
    Following Remark \ref{rem:ss2cocycles} and as described for the previous example, 
    we see that also in this case we can take as 
    2-cocycle any 2-form represented by an arbitrary skew-symmetric matrix. So, we can build the following 1+0 non-homogeneous HO:
    \begin{equation}
        \mathcal{A}=
        a
        \begin{pmatrix}
            1 & 0 & 0 \\
            0 & 1 & 0 \\
            0 & 0 & 1 \\
        \end{pmatrix}
        \partial_x +
        \begin{pmatrix}
            0 & u^3 & -u^2 \\
            -u^3 & 0 & u^1 \\
            u^2 & -u^1 & 0
        \end{pmatrix}
        +
        \begin{pmatrix}
            0 & f^{12} & f^{13} \\
            -f^{12} & 0 & f^{23} \\
            -f^{13} & -f^{23} & 0
        \end{pmatrix},
    \end{equation}
    where $\Set{u^{1},u^{2},u^3}$ is a basis for the
    (commuting) variables in the symmetric algebra, and $a\in\R$ is 
    an arbitrary constant.
    
\paragraph{The $\mathfrak{so}(n,\R)$ algebra.}
    \label{ex:son}
    Let us consider the general case of the $\mathfrak{so}(n,\R)$ Lie algebra.
    This algebra can be represented by the algebra of skew-symmetric
    $n\times n$ matrices:
    \begin{equation}
        \mathfrak{so}(n,\R) = 
        \Set{M \in M_{n,n}(\R) | M+M^T = O_n}.
    \end{equation}
    This algebra is generated by the following $n(n-1)/2$ matrices:
    \begin{equation}
        N_{i,j}=E_{i,j}-E_{j,i}, \quad 1\leq i<j\leq n,
    \end{equation}
    and has commutation relations:
    \begin{equation}
        [N_{i,j},N_{k,l}] =
        \delta_{j,k}N_{i,l}-\delta_{j,l}N_{i,k}
        -\delta_{i,k}N_{j,l}+\delta_{i,l}N_{j,k},
    \end{equation}
    see \eqref{eq:EijEkm}. The associated Killing matrix has the particularly simple form:
    \begin{equation}
        K = -(n+2)I_{n(n-1)/2},
    \end{equation}
    where $I_N$ is the identity $N\times N$ matrix. {
    Finally, as pointed out before, using Remark \ref{rem:ss2cocycles},  the generic 2-cocycle is then explicitly computed for $\mathfrak{so}(n+1,\mathbb{R})$.}

We will present the explicit example of $\mathfrak{so}(4,\mathbb{R})$ in the next subsection.  
\vspace{12pt}

We observe that the same procedure can be applied also to more ``exotic''
simple Lie algebras, such as the exceptional Lie algebras $\galg_2$, $\mathfrak{f}_4$,
$\mathfrak{e}_6$, $\mathfrak{e}_7$, and $\mathfrak{e}_8$. Since these algebras
are of dimension at least 14 the associated computations are rather cumbersome,
even though they follow trivially from known formulas, see {\cite[Chap IV]{Jacobson1962}}.
As a title of example, in \Cref{app:g2} we show the formulas for the exceptional
Lie algebra of the smallest dimension, namely $\galg_2$.

\subsection{Direct sums of Lie algebras}
\label{directsum-sub}

In this Section, we now discuss a rather trivial case, which despite its plainness it will
be the source of many more examples: the case of the direct sums of Lie algebras.

\begin{theorem}
    Let $\galg=\galg_1\oplus\galg_2$, where $\galg_k$ are Lie algebras and let $C_k$ be a Casimir 
    element of $\galg_k$. Then $C_k$ are also Casimir elements for $\galg$. In particular, if the Casimir $C_k$ are quadratic
    and non-degenerate, then the most general non-degenerate quadratic Casimir polynomial is given by:
    \begin{equation} 
        C=\alpha_1 C_1+ \alpha_2 C_2 
        + \sum_{i=1}^{\dim Z(\galg_1)}
        \sum_{j=1}^{\dim Z(\galg_2)}
        a_{ij} z^{i}_{(1)}z^{j}_{(2)},
        \label{eq:casdirsum}
    \end{equation}
    where
    \begin{equation}
        Z(\galg_k) = \Span\Set{z^1_{(k)},\ldots,
        z^{\dim Z(\galg_k)}_{(k)}}.
    \end{equation}
    \label{thm:dirsum}
\end{theorem}

This follows immediately from the linearity of the Lie bracket and that for
a direct sum we have $[\galg_1,\galg_2] = 0$ and that $Z(\galg_1\oplus\galg_2) = Z(\galg_1)\oplus Z(\galg_2)$. 
Note that if $Z(\galg_1)= \emptyset$ or $Z(\galg_2)= \emptyset$, then 
the off-diagonal elements in \eqref{eq:casdirsum} are not present.

A trivial example of direct sum are again the abelian Lie algebras $k\mathfrak{n}_{1,1}$ considered in \Cref{sss:abelian}. Indeed, it is possible to write $k\mathfrak{n}_{1,1} = \bigoplus_{i=1}^{k}\mathfrak{n}_{1,1}$, and then recover the results of \Cref{thm:abelian} through \Cref{thm:dirsum}.

The most relevant consequence of \Cref{thm:dirsum} is the following:

\begin{corollary}
    If a Lie algebra $\galg$ is such that either:
    \begin{itemize}
        \item $\galg = \mathfrak{s}\oplus (k \mathfrak{n}_{1,1})$ where $\mathfrak{s}$
            is semi-simple and $k\mathfrak{n}_{1,1}$ is abelian;
        \item $\galg = \mathfrak{s}_1\oplus \mathfrak{s}_2$ where $\mathfrak{s}_i$
            are semi-simple;
    \end{itemize}
    Then $\galg$ admits a non-degenerate quadratic Casimir polynomial.
    \label{cor:dirsum}
\end{corollary}

We complement this result with the following that we will use to
compute the 2-cocycles in the various examples:

\begin{lemma}
    Consider a Lie algebra $\galg=\galg_1\oplus\galg_2$ where $\mathfrak{g}_1$ and $\galg_2$ have
    structure constants $c^{ij}_k$, and $\gamma^{ij}_k$ respectively. Fix the dual basis of $\galg^*$ as:
    \begin{equation}
        \Lambda^1 \galg = \Span\Set{\theta_1,\ldots,\theta_n,
        \phi_{1},\ldots,\phi_{k}}=\Lambda^1\galg_1
        \oplus\Lambda^1 \galg_2.
    \end{equation}
    Then, a 2-cocycle
    $\omega \in \Lambda^2 \galg$ is of the following form:
    \begin{equation}
        \omega = \alpha_1 + \alpha_2 
        + \beta^{ij'}\theta_{i}\wedge\phi_{j'},
        \label{eq:2formdec}
    \end{equation}
    where $\alpha_i\in\Lambda^2\galg_i$ is a 2-cocycle
    of $\mathfrak{g}_i$, and the coefficients $\beta^{ij'}$ solve the linear system:
    \begin{equation}
        \beta^{ij'}c_{i}^{hl}=
        \beta^{ij'} \gamma_{j'}^{p'q'} =0.
        \label{eq:gammaeq}
    \end{equation}
    \label{lem:2cocyclessum}
\end{lemma}

\begin{proof}
    The most general 2-form $\omega\in\Lambda^2\mathfrak{g}$ has the form given
    in equation \eqref{eq:2formdec}. Setting $\ud \omega=0$ gives us the 2-cocycle 
    condition.  Taking the exterior derivative we have:
    \begin{equation}
        \ud\omega = 
        \ud\alpha_1 + \ud\alpha_2
        -\frac{\beta^{ij'}}{2} 
        \left(c_{i}^{hl} \theta_{h} \wedge \theta_{l} \wedge \phi_{j'}
        - \gamma_{j'}^{p'q'} \theta_{i} \wedge \phi_{p'} \wedge \phi_{q'}\right).
    \end{equation}
    Now, since $\galg$ is a direct sum we have $\ud\alpha_i \in \Lambda^3\galg_i$,
    so that $\ud \omega= 0$ implies $\ud\alpha_i = 0$. The conditions in \eqref{eq:gammaeq}
    follows from noting that the elements $\theta_{h} \wedge \theta_{l} \wedge \phi_{j'}$
    and $\theta_{i} \wedge \phi_{p'} \wedge \phi_{q'}$ are linearly independent in
    $\Lambda^3\galg$ and taking coefficients with respect to them.
\end{proof}

Let us now present some examples of this occurrences.

\paragraph{The algebras $\Sl(2,\R) \oplus (k\mathfrak{n}_{1,1})$.}
    In this example we consider the first case presented in \Cref{cor:dirsum}.
    With some care the procedure adapts to all simple Lie algebras.
    Let us consider the Lie algebra $\Sl(2,\R) \oplus (k\mathfrak{n}_{1,1})$, i.e. a trivial central extension of $\Sl(2,\R)$, with the same basis as in \Cref{semisimple-sub}.
    From \Cref{thm:dirsum} and \Cref{eq:carsl2}, since $Z(\Sl(2,\R))=\emptyset$ 
    we have the following quadratic Casimir element:
    \begin{equation}
        C = J^+J^- + J^- J^+   - 2(J^3)^2 + \sum_{i,j}^{k} a_{ij}e^{i}e^{j},
    \end{equation}
    where $a_{ij}$ are the elements of a symmetric matrix.

    Now, we show that in this case the 2-cocycle is:
    \begin{equation}
        \omega = \alpha_{\Sl(2,\R)} + \alpha_{k\mathfrak{n}_{1,1}}
        \label{eq:sl2ka2cocycle}
    \end{equation}
    where  $\alpha_{\Sl(2,\R)}$ and $\alpha_{k\mathfrak{n}_{1,1}}$ are the 2-cocycles of 
    $\Sl(2,\R)$ and $k\mathfrak{n}_{1,1}$ respectively, i.e.\ 
    an arbitrary $3\times 3$ skew-symmetric matrix and
    an arbitrary $k\times k$ skew-symmetric matrix. Indeed, denoting by $\Set{\theta_{+},\theta_{-},\theta_3,\phi_1,\ldots,\phi_{k}}$ the basis of $\Lambda\galg$, following \Cref{lem:2cocyclessum} we can consider the ``mixed'' 2-form:
    \begin{equation}
        \beta = \theta_{+}\wedge\sum_{j=1}^{k}\beta^{+,j}\varphi_{j}
        +\theta_{-}\wedge\sum_{j=1}^k\beta^{-,j}\varphi_{j}
        +\theta_3\wedge\sum_{j=1}^k\beta^{3\,j}\varphi_{j}.
    \end{equation}
    Noting that:
    \begin{equation}
        \ud (\theta_{\pm}\wedge\varphi_{j})=
        \pm\theta_{\pm}\wedge\theta_{3}\wedge\varphi_{j},
        \quad
        \ud (\theta_{3}\wedge\varphi_{j})=
        4\theta_{+}\wedge\theta_{-}\wedge\varphi_{j},
    \end{equation}
    we have that $\beta^{\pm,j}=\beta^{3,j}=0$ proving
    formula~\eqref{eq:sl2ka2cocycle}.

    These considerations shows that the 1+0 HO associated to 
    $\Sl(2,\R)\oplus(k\mathfrak{n}_{1,1})$ has the following block structure:
    \begin{equation}
        \mathcal{A} = 
        \begin{pmatrix}
            g_{\Sl(2,\R)} & O_{3,k}
            \\
            O_{k,3} & S_{k,k}
        \end{pmatrix}
        \partial_x
        +
        \begin{pmatrix}
            T_{\Sl(2,\R)} & O_{3,k}
            \\
            O_{k,3} & O_{k,k}
        \end{pmatrix}
        +
        \begin{pmatrix}
            F_{\Sl(2,\R)} & O_{3,n}
            \\
            O_{n,3} & F_{k\mathfrak{n}_{1,1}}
        \end{pmatrix}.
        \label{eq:hosl2kn11}
    \end{equation}
    It is possible to prove that the form of the 1+0 HO associated to
    the direct sum of the general special linear algebra with an abelian one, i.e.\ 
    $\Sl(n,\R)\oplus(k\mathfrak{n}_{1,1})$, is analogous to~\Cref{eq:hosl2kn11}.
    
\paragraph{The Lie algebra $\So(4,\mathbb{R})$.}

In this paragraph we consider the second case presented in \Cref{cor:dirsum}.
We do it with a nontrivial example, i.e.\ the Lie algebra $\So(4,\mathbb{R})$.
Indeed, it is known that  $\So(4,\mathbb{R})\cong 
\So(3,\mathbb{R})\oplus\So(3,\mathbb{R})$, see  \cite{snowin}. We observe that this
decomposition is not apparent immediately in the basis of $\So(4,\R)$ 
described in the previous \Cref{semisimple-sub}, but it is obtained through 
an isomorphism. Writing the bases as:
\begin{equation}
    \So(3,\mathbb{R})=\set{L_1^{(1)},L_2^{(1)},L_3^{(1)}},
    \quad
    \So(3,\mathbb{R}) =\set{L_1^{(2)},L_2^{(2)},L_3^{(2)}}.
    \label{eq:so4}
\end{equation}
Then, the Casimir element is the sum of the two Casimir elements:
\begin{equation}
    C_{\So(4,\R)} = 
    \alpha_1\left[(L_1^{(1)})^2+(L_2^{(1)})^2+(L_3^{(1)})^2\right]+
    \alpha_2\left[(L_1^{(2)})^2+(L_2^{(2)})^2+(L_3^{(2)})^2\right].
    \label{eq:so4cas}
\end{equation}
Moreover, through a direct computation, we solve the 
system~\eqref{eq:gammaeq} and get that the 2-cocycle is such that
the mixed part annihilates:
\begin{equation}
    \begin{aligned}
        \omega &=
        \alpha^{1,2}_1 \theta_1\wedge\theta_2 +
        \alpha^{1,3}_1 \theta_1\wedge\theta_3 +
        \alpha^{2,3}_1 \theta_2\wedge\theta_3
        \\
        &+\alpha^{1,2}_2 \phi_1\wedge\phi_2 +
        \alpha^{1,3}_2 \phi_1\wedge\phi_3 +
        \alpha^{2,3}_2 \phi_2\wedge\phi_3.
    \end{aligned}
    \label{eq:so42coc}
\end{equation}

So, we conclude this example by showing the explicit form
of the operator:
\begin{equation}
    \begin{aligned}
        \mathcal{A} &=
        \begin{pmatrix}
            a_1 &  \\
             & a_1 & \\
             &  & a_1 \\
             &&& a_2\\
             &&&& a_2\\
             &&&&& a_2
        \end{pmatrix}
        \partial_x +
        \begin{pmatrix}
            0 & u^3_{(1)} & -u^2_{(1)} \\
            -u^3_{(1)} & 0 & u^1_{(1)} \\
            u^2_{(1)} & -u^1_{(1)} & 0\\
            &&& 0 & u^3_{(2)} & -u^2_{(2)} \\
            &&&-u^3_{(2)} & 0 & u^1_{(2)} \\
            &&&u^2_{(2)} & -u^1 & 0\\
        \end{pmatrix}
        \\
        &+
        \begin{pmatrix}
            0 & f^{12}_{(1)} & f^{13}_{(1)} \\
            -f^{12}_{(1)} & 0 & f^{23}_{(1)} \\
            -f^{13}_{(1)} & -f^{23}_{(1)} & 0\\
            &&&0 & f^{12}_{(2)} & f^{13}_{(2)} \\
            &&&-f^{12}_{(2)} & 0 & f^{23}_{(2)} \\
            &&&-f^{13}_{(2)} & -f^{23}_{(2)} & 0\\
        \end{pmatrix}.
    \end{aligned}
    \label{eq:so4operatorfin}
\end{equation}
Moreover we observe that similar computations hold for other
decomposable low-dimensional semi-simple Lie algebras like 
$\So(2,2,\R)\cong \Sl(2,\R)\oplus \Sl(2,\R)$.

\paragraph{An example with non-trivial mixed elements.}

We conclude this subsection with an example where mixed terms are
present both in the Casimir polynomial and in the cocycles. Let us
consider the solvable Lie algebra $\mathfrak{s}_{4,6}=\Span_\R \set{e^1,e^2,e^3,e^4}$,
whose non-zero commutation relations are 
\begin{equation}
    \begin{array}{ccc}
         [e^2,e^3]=e^1, & [e^4,e^2]=e^2, & [e^4,e^3]=-e^3,
    \end{array}
    \label{eq:s46comm}
\end{equation}
see~\cite[\S 17.3]{snowin}. We have $Z(\mathfrak{s}_{4,6})=\Span_\R\set{e^1}$.
Moreover, this algebra has the non-degenerate quadratic Casimir $C^{(2)}_{\mathfrak{s}_{4,6}}=e^1e^4+e^2e^3$, so that its full non-degenerate quadratic Casimir element
is given by:
\begin{equation}
    C_{\mathfrak{s}_{4,6}} = a_1 {(e^1)}^2 + a_2 C^{(2)}_{\mathfrak{s}_{4,6}}
    =a_1 {(e^1)}^2 + a_2(e^1e^4+e^2e^3).
\end{equation}
Moreover, from a direct computation we get that the space of the 2-cocycles
is 3-dimensional, and the most generic 2-cocycle has the following shape:
\begin{equation}
    \omega_{\mathfrak{s}_{4,6}} = 
    \Omega^{23} \theta_2\wedge\theta_3 + \Omega^{24}\theta_2\wedge\theta_4
    + \Omega^{34}\theta_3\wedge\theta_4.
\end{equation}
The explicit form of the operator associated to this Lie algebra will be
given in \Cref{sec4}.

Consider now the direct sum $\mathfrak{g}_{4,6}^{(k)} = 
\mathfrak{s}_{4,6}\oplus (k\mathfrak{n}_{1,1})$,
where the abelian summand is generated by the set $\set{f^1,\ldots,f^k}$. Then,
clearly:
\begin{equation}
    Z(\mathfrak{g}_{4,6}^{(k)}) 
    = 
    \Span_\R\set{e^1} \oplus k \mathfrak{n}_{1,1}.
\end{equation}
Based on \Cref{thm:dirsum} we have that $\mathfrak{g}_{4,6}^{(k)}$ admits
the following non-degenerate quadratic Casimir:
\begin{equation}
    C_{\mathfrak{g}_{4,6}^{(k)}} = 
    a_1 {(e^1)}^2 + a_2(e^1e^4+e^2e^3) 
    + a_{ij} f^{i}f^{j}
    + b_i e^1f^i,
    \label{eq:Cg46k}
\end{equation}
where $a_{ij}=a_{ji}$ and $b_i$ are arbitrary constants. It is clear that 
\Cref{eq:Cg46k} has non-trivial mixed elements. In the same way we see that 
the basis of the 2-cocycles of $\mathfrak{g}_{4,6}^{(k)}$ has mixed elements. 
Indeed, let us take the generic mixed element of 
$\mathfrak{g}_{4,6}^{(k)}$, i.e.\ 
$\omega_\text{mix} = \beta^{ij'}\theta_i\wedge\phi_{j'}$. Then we have:
\begin{equation}
    \begin{array}{ll}
         \ud (\theta_1 \wedge \phi_i) = - \theta_2\wedge\theta_3\wedge\phi_i, 
         & 
         \ud (\theta_2 \wedge \phi_i) =  \theta_2\wedge\theta_4\wedge\phi_i, 
         \\[3pt]
         \ud (\theta_3 \wedge \phi_i) = - \theta_3\wedge\theta_4\wedge\phi_i, 
         &
         \ud (\theta_4 \wedge \phi_i) = 0,
    \end{array}
\end{equation}
and since the last exterior derivative is identically zero we have that the 
terms $\beta^{4 j'}$ \emph{are free}. That is, the mixed terms are present,
and the the most general 2-cocycle of $\mathfrak{g}_{4,6}^{(k)}$ is:
\begin{equation}
    \begin{aligned}
        \omega_{\mathfrak{g}_{4,6}^{(k)}}
        &=
        \Omega^{23} \theta_2\wedge\theta_3 
        + \Omega^{24}\theta_2\wedge\theta_4
        + \Omega^{34}\theta_3\wedge\theta_4
        \\
        &+ \beta^{4j'}\theta_4\wedge\phi_{j'}
        +\Upsilon^{i'j'}\phi_{i'}\wedge\phi_{j'}.
    \end{aligned}
\end{equation}

So, in the end these considerations shows that the 1+0 HO associated to 
$\mathfrak{g}_{4,6}^{(k)}$ has the following block structure:
\begin{equation}
    \mathcal{A} = 
    \begin{pmatrix}
        g_{\mathfrak{s}_{4,6}} & M_{4,k}
        \\
        M_{4,k}^{T} & S_{k,k}
    \end{pmatrix}
    \partial_x
    +
    \begin{pmatrix}
        T_{\mathfrak{s}_{4,6}} & O_{2,k}
        \\
        O_{k,2} & O_{k,k}
    \end{pmatrix}
    +
    \begin{pmatrix}
        F_{\mathfrak{s}_{4,6}} & \digamma_{4,n}
        \\
        -\digamma_{4,n}^T & F_{k\mathfrak{n}_{1,1}}
    \end{pmatrix},
        \label{eq:hos46kn11}
\end{equation}
where
\begin{equation}
    \begin{array}{ll}
         g_{\mathfrak{s}_{4,6}} 
         =
         \begin{pmatrix}
             & & & g_2
             \\
             && g_2
             \\
             & g_2
             \\
             g_2 & & & g_1
         \end{pmatrix},
         & M_{2,k} = 
         \begin{pmatrix}
             \\
             \\
             \\
             h_1 & h_2 & \ldots & h_k
         \end{pmatrix},
         \\[24pt]
         T_{\mathfrak{s}_{4,6}} 
         =
         \begin{pmatrix}
             0
             \\
             && u^{1} & u^{2}
             \\
             &-u^{1} & 0 & -u^{3}
             \\
             & -u^{2} & u^{3} &0
         \end{pmatrix},
         & 
         F_{\mathfrak{s}_{4,6}}
         =
         \begin{pmatrix}
             0
             \\
             &  0 & \Omega^{23} & \Omega^{24}
             \\
             & -\Omega^{23} & 0 & \Omega^{34}
             \\
             & -\Omega^{24} & -\Omega^{34} & 0
         \end{pmatrix},
         \\[24pt]
         \digamma_{4,n}
         =
         \begin{pmatrix}
             \\
             \\
             \\
             \beta^{41} & \beta^{42} & \ldots & \beta^{4k}
         \end{pmatrix}.
    \end{array}
\end{equation}
where $g_i$, $h_i$ are arbitrary constants, and
$S_{k,k}$ is an arbitrary symmetric matrix, while 
$F_{k\mathfrak{n}_{1,1}}$ is an arbitrary skew-symmetric matrix.
It is possible to prove that the form of the 1+0 HO associated to
the direct sum of two copies of $\mathfrak{s}_{4,6}$, i.e.\ 
$\mathfrak{s}_{4,6}\oplus\mathfrak{s}_{4,6}$, is similar 
to~\eqref{eq:hos46kn11}.

\subsection{Two-step nilpotent Lie algebras}
\label{sss:nilpot}

In this subsection, we prove some general results on nilpotent
Lie algebras admitting a non-degenerate quadratic Casimir element.
Let us note that this is far from being general, since many nilpotent
Lie algebras does not have non-degenerate Casimir. As a very simple 
example of this, we consider the Heisenberg algebra 
$\mathfrak{n}_{3,1}=\mathfrak{h}(1)=\Span\Set{e^1,e^2,e^3}$ (see~\cite[Section 16.3]{snowin}), whose
only non-zero commutation relation is $[e^2,e^3]=e^1$. This algebra is nilpotent and it
admits a single linear Casimir invariant $C_{\mathfrak{n}_{3,1}}=e^1$, whose square is degenerate.

We focus our attention on two-step nilpotent Lie algebras, i.e. those
nilpotent Lie algebras $\mathfrak{n}$ such that $\mathfrak{n}^2=[\mathfrak{n},\mathfrak{n}]=0$. For such Lie algebras we have the following characterisation:

\begin{lemma}[\cite{Scheuneman1966PhD}]
    Let $\mathfrak{n}$ be a two-step nilpotent Lie algebra.
    Then there exists a basis $\Set{e^1,\dots e^n,f^1,\dots , f^m}$
    such that the commutation relation have the following form:
    \begin{align}
        \label{condsc}[e^i,e^j]=c_s^{ij}f^s, \quad [e^i,f^j]=[f^i,f^j]=0.
    \end{align}
    \label{lem:2stepstruct}
\end{lemma}
Then, we have the following result to build Casimir elements for 2-step nilpotent Lie algebras:

\begin{theorem}\label{thm:twostepnilpotent}
    Let $\mathfrak{n}$ be a two-step nilpotent Lie algebra of dimension $2n$. If $\mathfrak{n}=\{e^1,\dots e^n,f^1,\dots , f^n\}$ and { the coefficients $c_s^{ij}$ in  formula \eqref{condsc} define the structural constants of a Lie algebra that admits a nontrivial quadratic Casimir, then $\mathfrak{n}$ admits a nontrivial quadratic Casimir.}  
\end{theorem}

\begin{proof}

    Let $\mathfrak{g}=\{g^1,\dots g^n\}$ a Lie algebra structure whose structural constants are $c_i^{jk}$ and let us assume $C_{\mathfrak{g}}$ is a Casimir element of $\mathfrak{g}$:
    \begin{equation}
        C_{\mathfrak{g}}=\frac{a_{ij}}{2}\left(g^ig^j+g^jg^i\right).
    \end{equation}
    We define 
    \begin{equation}\label{344}
        C_{\mathfrak{n}}=\frac{a_{ij}}{2}\left(e^if^j+f^ie^j\right)+b_{ij}f^if^j,
    \end{equation}
    where $b_{ij}\in\mathbb{R}$ are arbitrary and we use the simplification $f^if^j=\frac{1}{2}\left(f^if^j+f^jf^i\right)$ due to its commutativity property.

    It now remains to prove that $C_{\mathfrak{n}}$ is a Casimir element of $\mathfrak{n}$. First, we observe that $[C_\mathfrak{n},f^k]$ vanishes for every $k=1,\dots , n$ due to \eqref{condsc}. Now, 
    \begin{align}
       \begin{split} [C_\mathfrak{g},e^k]&=C_\mathfrak{g}e^k-e^kC_\mathfrak{g}\\
        &=\left(\frac{a_{ij}}{2}\left(e^if^j+f^ie^j\right)+b_{ij}f^if^j\right)e^k+\\
        &\hphantom{ciao}-e^k\left(\frac{a_{ij}}{2}\left(e^if^j+f^ie^j\right)+b_{ij}f^if^j\right)\\
        &=\frac{a_{ij}}{2}\left(e^if^j+f^ie^j\right)e^k+\\
        &\hphantom{ciao}-e^k\frac{a_{ij}}{2}\left(e^if^j+f^ie^j\right)\\
        &=\frac{a_{ij}}{2}[\left(e^je^k-e^ke^j\right)f^i+\left(e^ie^k-e^ke^i\right)f^j]\\
        &=\frac{a_{ij}}{2}\left(c^{jk}_lf^lf^i+c^{ik}_lf^lf^j\right)\\
        &=\frac{1}{2}\left(a_{ij}c^{jk}_l+a_{lj}c^{jk}_i\right)f^if^l\\
        &=0.\end{split}
    \end{align}
\end{proof}

\paragraph{A two-step nilpotent algebra with $\Sl(2,\mathbb{R})$ commutation relations.}
Let us consider the algebra $\tilde{\mathfrak{n}}$, explicitly, 
    \begin{equation}
       \tilde{ \mathfrak{n}}=\{e^1,e^2,e^3,f^1,f^2,f^3\},
    \end{equation}
    whose non-zero commutation relations are
    \begin{equation}
        [e^1,e^2]=f^2, \quad [e^1,e^3]=-f^3, \quad [e^2,e^3]=-f^1.
    \end{equation}
    These constants $c^{ij}_k$ define the structure constants of $\mathfrak{su}(1,1)$, so from Theorem \ref{thm:twostepnilpotent} $\tilde{\mathfrak{n}}$ admits a non-degenerate Casimir element of type \eqref{344}, i.e. 
    \begin{equation}
        C=\frac{1}{\alpha} (e^1f^1-e^2f^3-e^3f^2),
    \end{equation} where $\alpha$ is an arbitrary constant. The resulting associated operator is
    \begin{align}
        \begin{split}\tilde{\mathcal{A}}&= \begin{pmatrix}
                \beta^{11}&\beta^{12}&\beta^{13}&\alpha&0&0\\\beta^{12}&\beta^{22}&\beta^{23}&0&0&-\alpha\\\beta^{13}&\beta^{23}&\beta^{33}&0&-\alpha&0\\\alpha&0&0&0&0&0\\0&0&-\alpha&0&0&0\\0&-\alpha&0&0&0&0
            \end{pmatrix}\partial_x+\\&\hphantom{cioaicoaio}+\begin{pmatrix}
            0 & f^{12}+u^4 & f^{13}-u^6 & f^{14} & f^{15} & f^{16} \\
 -f^{12}-u^4 & 0 & f^{23}-u^3 & f^{24} & f^{25} & f^{26} \\
 u^6-f^{13} & u^4-f^{23} & 0 & f^{34} & f^{14}-f^{26} & f^{36} \\
 -f^{14} & -f^{24} & -f^{34} & 0 & 0 & 0 \\
 -f^{15} & -f^{25} & f^{26}-f^{14} & 0 & 0 & 0 \\
 -f^{16} & -f^{26} & -f^{36} & 0 & 0 & 0
        \end{pmatrix}\end{split}
    \end{align}

\begin{remark}
    We remark that Theorem \ref{thm:twostepnilpotent} gives sufficient but not necessary conditions. Indeed, a counterexample is the algebra $\mathfrak{n}_{6,1}$ \cite{snowin} 
    \begin{equation}
        \text{span}\{n^1,\dots, n^6 \}
    \end{equation}
    with the non-zero commutation relations given by 
    \begin{equation}\label{relas}
        [n^4,n^5]=n^2, \quad [n^4,n^6]=n^3, \quad [n^5,n^6]=n^1.
    \end{equation}
    This Lie algebra is 2-step nilpotent and admits the following Casimir invariants
    \begin{equation}
        C_{1}=n^1, \quad C_2=n^2, \quad C_3=n^3, \quad C^*=n^1n^4+n^2n^6-n^3n^5,
    \end{equation}
    so that we can build a non-degenerate quadratic Casimir element. 

    However, we observe that the previous Theorem is not applicable here: the relation \eqref{relas} are not the structure constant of a Lie algebra. This apparent inconsistency in our result is solved considering that the algebra $\mathfrak{n}_{6,1}$ is isomorphic to the algebra $\tilde{\mathfrak{n}}$ considered in the previous example.
 
\end{remark}
\paragraph{Further considerations on $k$-steps nilpotent Lie algebras.}
        We stress that there exist also Lie algebra structures with non-degenerate quadratic Casimir elements which are nilpotent with more than two steps. As an example, let us consider the 3-step nilpotent Lie algebra   $\mathfrak{n}_{5,2}=\Span\Set{e^1,\ldots,e^5}$ with commutation relations 
        \begin{equation}
            [e^3,e^4]=e^2, \qquad [e^3,e^5]=e^1, \qquad [e^4,e^5]= e^3.
        \end{equation} It admits the non-degenerate quadratic Casimir element $(e^3)^2 +2e^2e^5 -2e^1e^4$. We will show the operator associated to this Lie algebra in the next Section.

\section{Isomorphism classes of operators associated to low-dimensional Lie algebras}\label{sec4}

In this Section, we present a complete description in low dimensions of the non-homogeneous hydrodynamic type HOs by making use of the isomorphism classes of the associated Lie algebra structures. We stress that our results give all the possible cases of such HOs with non-degenerate leading coefficient. The term ``low'' here is referring to the theory of real Lie algebras, so that is considered up to $n=6$, i.e.\ up to the dimension where Lie algebras are \emph{completely classified} (see below for more details on this
classification).

Strictly speaking, our construction cannot be considered a proper \emph{classification} of operators. Indeed, we do not apply further transformations of variables to the 2-cocycles arising from the previous section. To do this, one should consider such linear transformations that preserve the structure constants of the Lie algebras (i.e. Lie algebra automorphisms) and apply them to the 2-cocycle in order to further reduce their degrees of freedom. In general, this represents a difficult task, indeed one need to build the full group of automorphisms of a given Lie algebra. Computationally, this amounts to characterise the intersection of a huge number of quadrics which is a very complicated problem in algebraic geometry. For our purposes, such a procedure is inessential because as it will be discussed in the example of the Korteweg-de Vries equation it can be useful to keep the free parameters on the components of the 2-cocycle.    

In what follows, we make use of the classification of \emph{real} Lie algebras in 
dimension $n\leq 6$  as reported in \cite[Chapters 15--19]{snowin}. The reader is 
advised that the order of such a  list is different from other presented in the literature, 
but has the advantage to be more  refined in the specific properties of the algebras. 
The topic of classification of Lie algebras is a very old topic, dating back to the 
introduction of Lie algebras themselves. A complete history of the this topic is given 
in~\cite{Popovych2003}. Here we will just resume the foundations that we need to present
the classification of Lie algebras up to dimension $6$. By Levi--Malchev theorem the 
classification of Lie algebras boils down to classify three different
types: the semi-simple algebras, the solvable algebras, and the semidirect sums of 
solvable and semisimple algebras. As recalled in~\Cref{sec3} semi-simple 
Lie algebras of finite dimension are classified by Cartan's criterion by their Killing 
form. Following  Cartan~\cite{Cartan1909}, semi-simple algebras are divided in four infinite 
classed and a finite number of exceptional cases. Without additional assumptions on the 
Lie algebra, Sophus Lie himself classified all complex Lie algebras of dimension less than 
$4$~\cite{LieTransformationGroups}. Then, Bianchi gave the first classification of  3-dimensional real Lie algebras~\cite{Bianchi1918}. Later, 
Mubarakzyanov~\cite{Mubarakzyanov1963a,Mubarakzyanov1963b,Mubarakzyanov1963c} gave a 
complete classification of 4- and 5-dimensional real Lie algebras, and  
in~\cite{Mubarakzyanov1966} classified all 6-dimensional real Lie algebras with  one 
linearly independent non-nilpotent element. Finally, more recently Turkowski 
in~\cite{Turkowski1990} completed Mubarakzyanov’s classification of 
6-dimensional solvable Lie algebras, by classifying real Lie algebras of dimension 6 that 
contain four-dimensional nilradical.

Before presenting the obtained results, we want to remark that \emph{almost all} the arising HOs come from classes of Lie algebras we discussed in the previous section. This underlines how the approach we built in Section \ref{sec3} is quite general, even though some seemingly isolated \emph{unprecedented} examples can arise. A resume of the results
can be found in \Cref{tab:general}.

\begin{table}\centering
\begin{tabular}{cccc}
    Dimension & Operator &  Algebra & Structure \\
    \hline
    2& $\mathcal{A}_{2,1}$ & $2\mathfrak{n}_{1,1}$ & Abelian \\
    \hline
    3& $\mathcal{A}_{3,1}$ & $3\mathfrak{n}_{1,1}$ & Abelian \\
    \hline
    3& $\mathcal{A}_{3,2}$ & $\mathfrak{sl}(2,\mathbb{R})$ & Simple \\
    \hline
     3& $\mathcal{A}_{3,3}$ & $\mathfrak{so}(3,\mathbb{R})$ & Simple \\
    \hline
    4& $\mathcal{A}_{4,1}$ & $4\mathfrak{n}_{1,1}$ & Abelian \\
    \hline
    4& $\mathcal{A}_{4,2}$ & $\mathfrak{s}_{4,6}$ & Solvable \\
    \hline
    4& $\mathcal{A}_{4,3}$& $\mathfrak{s}_{4,7}$ & Solvable \\
    \hline
    4& $\mathcal{A}_{4,4}$ & $\mathfrak{sl}(2,\mathbb{R})\oplus \mathfrak{n}_{1,1}$ & Direct sum \\
    \hline
    4& $\mathcal{A}_{4,5}$ & $\mathfrak{so}(3,\mathbb{R})\oplus \mathfrak{n}_{1,1}$ & Direct sum \\
    \hline
    5& $\mathcal{A}_{5,1}$ & $5\mathfrak{n}_{1,1}$ & Abelian \\
    \hline
     5&$\mathcal{A}_{5,2}$ & $\mathfrak{sl}(2,\mathbb{R})\oplus2\mathfrak{n}_{1,1}$ & Direct sum \\
     \hline
     5& $\mathcal{A}_{5,3}$& $\mathfrak{so}(3,\mathbb{R})\oplus2\mathfrak{n}_{1,1}$ & Direct sum \\
    \hline
    5& $\mathcal{A}_{5,4}$& $\mathfrak{s}_{4,6}\oplus\mathfrak{n}_{1,1}$ & Direct sum \\
    \hline
     5& $\mathcal{A}_{5,5}$& $\mathfrak{s}_{4,7}\oplus\mathfrak{n}_{1,1}$ & Direct sum \\
    \hline
    5&$\mathcal{A}_{5,6}$ & $\mathfrak{n}_{5,2}$ & 3-Step Nilpotent \\
    \hline
     6&$\mathcal{A}_{6,1}$ & $6\mathfrak{n}_{1,1}$ & Abelian \\
     \hline
    6& $\mathcal{A}_{6,2}$ & $\mathfrak{sl}(2,\mathbb{R})\oplus3\mathfrak{n}_{1,1}$ & Direct sum \\
    \hline
    6& $\mathcal{A}_{6,3}$ & $\mathfrak{so}(3,\mathbb{R})\oplus3\mathfrak{n}_{1,1}$ & Direct sum \\
    \hline
    6& $\mathcal{A}_{6,4}$ & $\mathfrak{sl}(2,\mathbb{R})\oplus\mathfrak{sl}(2,\mathbb{R})\cong\mathfrak{so}(2,2,\mathbb{R})$ & Direct sum \\
    \hline
    6& $\mathcal{A}_{6,5}$ & $\mathfrak{so}(3,\mathbb{R})\oplus\mathfrak{so}(3,\mathbb{R})\cong\mathfrak{so}(4,\mathbb{R})$ & Direct sum \\
    \hline
    6& $\mathcal{A}_{6,6}$ & $\mathfrak{sl}(2,\mathbb{R})\oplus\mathfrak{so}(3,\mathbb{R})$ & Direct sum \\
    \hline
    6& $\mathcal{A}_{6,7}$ & $\mathfrak{s}_{4,6}\oplus 2\mathfrak{n}_{1,1}$ & Direct sum \\
    \hline
    6& $\mathcal{A}_{6,8}$ & $\mathfrak{s}_{4,7}\oplus 2\mathfrak{n}_{1,1}$ & Direct sum \\
    \hline
    6& $\mathcal{A}_{6,9}$ & $\mathfrak{n}_{5,2}\oplus\mathfrak{n}_{1,1}$ & Direct sum \\
    \hline
    6& $\mathcal{A}_{6,10}$ & $\mathfrak{n}_{6,1}$ & 2-Step Nilpotent\\
    \hline
    6&$\mathcal{A}_{6,11}$ & $\mathfrak{s}_{6,162}$ & Solvable\\
    \hline
    6& $\mathcal{A}_{6,12}$ & $\mathfrak{s}_{6,163}$ & Solvable\\
    \hline
    6& $\mathcal{A}_{6,13}$ & $\mathfrak{s}_{6,164}$ & Solvable\\
    \hline
    6& $\mathcal{A}_{6,14}$ & $\mathfrak{s}_{6,165}$ & Solvable\\
    \hline
    6& $\mathcal{A}_{6,15}$ & $\mathfrak{s}_{6,166}$ & Solvable\\
    \hline
    6& $\mathcal{A}_{6,16}$ & $\mathfrak{s}_{6,167}$ & Solvable\\
    \hline
    6& $\mathcal{A}_{6,17}$ & $\mathfrak{so}(1,3,\mathbb{R})$ & Simple\\
    \hline
    6& $\mathcal{A}_{6,18}$ & $\mathfrak{sl}(3,\mathbb{R})\ltimes 3\mathfrak{n}_{1,1}$ & Levi decomposable\\
    \hline
    \end{tabular}
    \caption{Complete list of non-homogeneous hydrodynamic type $1+0$ HOs
    associated to the isomorphism classes of low dimensional Lie algebras,
    together with the relevant properties of the Lie algebra.} 
    \label{tab:general}
\end{table}

\subsection{Lie algebras of dimensions $n = 2, 3$}  

In dimension $2$, there are no non-abelian Lie algebra structures compatible with non-degenerate scalar products. This can be shown by direct computations or simply observing that  the system $\eta^{is}c^{jk}_s+\eta^{js}c^{ik}_s=0$ is a homogeneous linear system in the unknown constants $c^{ij}_k$ whose matrix of coefficient has maximal rank, i.e. it only admits the trivial solution $c^{ij}_k=0$.  This means that the associated operator is reduced to a constant form whose ultralocal term is only given by a 2-cocycle $f=f^{12}\frac{\partial }{\partial u^1}\wedge \frac{\partial}{\partial u^2}$ (the related Lie algebra is $2\mathfrak{n}_{1,1}$, where all the structure constants vanish). Note that for $n=2$, every skew-symmetric constant 2-form is a 2-cocycle. 

The general form of the operator is then the following one
\begin{equation}
    \mathcal{A}_{2,1}=\begin{pmatrix}a^{11}&a^{12}\\a^{12}&a^{22}\end{pmatrix}\partial_x+\begin{pmatrix}
        0&f^{12}\\-f^{12}&0
    \end{pmatrix},
\end{equation}where $a^{ij},f^{12}$ are arbitrary constants.

In the $3$-dimensional case, only 3 non-degenerate operators arise. The first one is of course given by the abelian algebra $3\mathfrak{n}_{1,1}$. The general structure of the operator is
\begin{equation}\label{31}
  \mathcal{A}_{3,1}=  \begin{pmatrix}
 a^{11} & a^{12} & a^{13} \\
 a^{12} &a^{22} & a^{23} \\
 a^{13} & a^{23} & a^{33} \\
\end{pmatrix}
\partial_x +
\begin{pmatrix}
 0 & f^{12} & f^{13} \\
 -f^{12} & 0 & f^{23} \\
 -f^{13} & -f^{23} & 0 \\
\end{pmatrix} , 
\end{equation}
where all the coefficients $a^{ij}$ and $f^{ij}$ are real and arbitrary constants.

The second case is $\mathfrak{su}(1,1)$\footnote{This Lie algebra is also indicated by $\mathfrak{sl}(2,\mathbb{R})$ in the Montreal notation.}, whose compatible scalar product depends on one free parameter (as the dimension on the space of quadratic Casimirs) and admits the following form
\begin{equation}\label{32}
  \mathcal{A}_{3,2}=  \begin{pmatrix}
 0 & 0 & \alpha \\
 0 & \frac{\alpha}{2} & 0 \\
 \alpha & 0 & 0 \\
\end{pmatrix}
\partial_x +
\begin{pmatrix}
 0 & u^1+f^{12} & -2 u^2+f^{13} \\
 -u^1-f^{12} & 0 & u^3+f^{13} \\
 2 u^2-f^{13} & -u^3-f^{23} & 0 \\
\end{pmatrix},
\end{equation}
with $\alpha,f^{ij}$ real constants.

We recall that the algebra $\mathfrak{su}(1,1)$ is semi-simple. Moreover, it is in the same isomorphism class as $\mathfrak{sl}(2,\mathbb{R})$, being the split form of the Dynkin diagram associated to $A_2$ (see \cite{Kirillov2008}).  We refer to subsection \ref{semisimple-sub} for further details on $\mathfrak{sl}(2,\mathbb{R})$ and its extensions. 

Finally, the Lie algebra $\mathfrak{so}(3,\mathbb{R})$ gives arise to 
\begin{equation}\label{33}
\mathcal{A}_{3,3}=\begin{pmatrix}
 \alpha & 0 & 0 \\
 0 & \alpha & 0 \\
 0 & 0 & \alpha \\
\end{pmatrix}
\partial_x +
\begin{pmatrix}
 0 & u^3+f^{12} & -u^2+f^{13} \\
 -u^3-f^{12} & 0 & u^1+f^{23} \\
 u^2-f^{13} & -u^1-f^{23} & 0 \\
\end{pmatrix},
\end{equation}
where $\alpha,f^{ij}$ are arbitrary real constants.  

We recall that the algebra $\mathfrak{so}(3,\mathbb{R})$ is semi-simple, and it is the compact form of the Lie algebra associated to the Dynkin diagram $A_2$ (see \cite{Kirillov2008}).  We refer to subsection \ref{semisimple-sub} for further details. 

A key example in the theory of integrable systems and nonlinear wave equations is given by the KdV equation. This turns out to be  bi-Hamiltonian with non-homogeneous structures when considered as a quasilinear system. We show further details in the following paragraph.

\paragraph{The Korteweg-de Vries equation.}\label{kdvsection}

We now consider the well-known KdV equation 
\begin{equation}
    u_t=6uu_x+u_{xxx}.
\end{equation}

First, we introduce new field variables $u^1=u,u^2=u_x$ and $u^3=u_{xx}$, writing accordingly the equation as a non-homogeneous quasilinear system. As described in \cite{mokhov98:_sympl_poiss},  Mokhov found a second transformation of coordinates on the 3-dimensional manifold of the field variables 
\begin{gather}\label{change}
\begin{split}
u^1&=\frac{w^1-w^3}{\sqrt{2}},\hspace{5ex} u^2=w^2,\hspace{5ex} u^3=\frac{w^1+w^3}{\sqrt{2}}+\left(w^1-w^3\right)^2 \,,
\end{split} 
\end{gather}
 which is also known as \emph{local quadratic unimodular change}. Using \eqref{change} into the obtained system,
 the original form of the KdV equation is mapped into 
\begin{equation}\label{kdvsys2}
    \begin{cases}
    w^1_t=-\dfrac{1}{2}\left(w^1-w^3\right)_x+w^2\left(w^1-w^3\right)+\dfrac{1}{\sqrt{2}}w^2\\
    w^2_t=\left(w^1-w^3\right)^2+\dfrac{1}{\sqrt{2}}\left(w^1+w^3\right)\\
    w^3_t=-\dfrac{1}{2}\left(w^1-w^3\right)_x+w^2\left(w^1-w^3\right)-\dfrac{1}{\sqrt{2}}w^2    \end{cases} \,.
\end{equation}
The latter turns out to be  a bi-Hamiltonian system with respect to two non-homogeneous operators of hydrodynamic type
\begin{subequations}\label{kdvvv}
\begin{align}
&\label{opkdv1}\mathcal{A}=\begin{pmatrix}
1&0&0\\0&-1&0\\0&0&-1
\end{pmatrix}\partial_x+\begin{pmatrix}
0&-2w^3&2w^2\\2w^3&0&2w^1\\-2w^2&-2w^1&0
\end{pmatrix}\, ,\\[2ex]
&\mathcal{B}=\frac{1}{2}\begin{pmatrix}
1&0&1\\0&0&0\\1&0&1
\end{pmatrix}\partial_x+\begin{pmatrix}
0&w^1-w^3+\frac{1}{\sqrt{2}}&0\\
w^3-w^1-\frac{1}{\sqrt{2}}&0&w^3-w^1+\frac{1}{\sqrt{2}}\\0&w^1-w^3-\frac{1}{\sqrt{2}}&0
\end{pmatrix}\, .
\end{align}
\end{subequations}
and the Hamiltonian functionals 
\begin{align}
    &H_{\mathcal{A}}=-\frac{1}{2}\int{\left((w^1-w^3)^2-\sqrt{2}(w^1+w^3)\right)\, dx}, \\ &H_{\mathcal{B}}=\int{\left((w^1)^2-(w^2)^2-(w^3)^2\right)\, dx}
\end{align}
such that \eqref{kdvsys2} reads as
\begin{equation}
    \label{biham}w^i_t=\mathcal{A}^{ij}\frac{\delta H_{\mathcal{A}}}{\delta w^j}=\mathcal{B}^{ij}\frac{\delta H_{\mathcal{B}}}{\delta w^j}.
\end{equation}
As Magri showed in \cite{Magri}, finding a pair of HOs $\mathcal{A}$ and $\mathcal{B}$ such that 
a system is written as in \eqref{biham} is strictly connected to the integrability (and the infinite number of symmetries and conserved quantities in involution) of the evolutionary system.  In spite this property is not a novelty for the KdV equation, such an example let us stress that in these new coordinates the bi-Hamiltonian property is given in terms of operators whose structure is the same, i.e. of non-homogeneous hydrodynamic type. We finally point out that the second operator $\mathcal{B}$ does not satisfy the non-degeneracy condition, that is its leading coefficient is a degenerate matrix.

\medskip

Before concluding, let us now observe that the first operator is non-degenerate and is the same shown for the 3-waves systems \eqref{3wav}. We then have that 

\begin{proposition}Through a linear change of variables, the operator $\mathcal{A}$ is mapped into $\mathcal{A}_{3,2}$ for $\alpha=-\frac{1}{2}$ and the 2-cocycle $f$ is identically zero.\end{proposition} 
\begin{proof}Let $\{w^1,w^2,w^3\}$ be a basis of $\mathfrak{g}_1$ such that
\[
[w^1,w^2] = -2w^3,\quad [w^1,w^3]=2w^2,\quad [w^2,w^3]=2w^1,
\]
associated to the Poisson tensor in \eqref{opkdv1}.
Considering the following change of coordinates
\[
\left\{
\begin{array}{lll}
u^1 &=& -w^1 + \frac{\sqrt{3}}{2}w^2-\frac{1}{2}w^3\\
u^2 &=& \frac{\sqrt{3}}{2}w^1 -w^2\\
u^3 &=& w^1 -\frac{\sqrt{3}}{2}w^2-\frac{1}{2}w^3
\end{array}
\right.
\]
we obtain the isomorphism $\Psi: \mathfrak{g}_1 \to \mathfrak{su}(1,1)$, mapping the operator $\mathcal{A}$ into
$$\tilde{\mathcal{A}}=\left(
\begin{array}{ccc}
 0 & 0 & -\frac{1 }{2} \\
 0 & -\frac{1 }{4} & 0 \\
 -\frac{1 }{2} & 0 & 0 \\
\end{array}
\right)\partial_x +
\left(
\begin{array}{ccc}
 0 & u^1 & -2 u^2 \\
 -u^1 & 0 & u^3 \\
 2u^2 & -u^3 & 0 \\
\end{array}
\right).$$
\end{proof}

\subsection{Lie algebras of dimension $n = 4$}

In the $4$-dimensional case, five non-degenerate operators arise. As usual,
the first one we consider is the one associated to the abelian Lie algebra 
of dimension $4$ (namely $4\mathfrak{n}_{1,1}$). It gives rise to the constant operator:
\begin{equation}
    \mathcal{A}_{4,1}=
   \left(
\begin{array}{cccc}
 a^{11} & a^{12} & a^{13} & a^{14} \\
 a^{12} & a^{22} & a^{23} & a^{24} \\
 a^{13} & a^{23} & a^{33} & a^{34} \\
 a^{14} & a^{24} & a^{34} & a^{44} \\
\end{array}
\right)\partial_x +\left(
\begin{array}{cccc}
 0 & f^{12} & f^{13} & f^{14} \\
 -f^{12} & 0 & f^{23} & f^{24} \\
 -f^{13} & -f^{23} & 0 & f^{34} \\
 -f^{14} & -f^{24} & -f^{34} & 0 \\
\end{array}
\right) ,
\end{equation}
where $a^{ij},f^{ij}$ are real.

Moreover, using the cited classification we obtain respectively by solvable 
Lie algebras $\mathfrak{s}_{4,6}$ and $\mathfrak{s}_{4,7}$:
\begin{equation}
    \mathcal{A}_{4,2}=   \left(
\begin{array}{cccc}
 0 & 0 & 0 & \alpha \\
 0 & 0 & -\alpha & 0 \\
 0 & -\alpha & 0 & 0 \\
 \alpha & 0 & 0 & \beta \\
\end{array}
\right)\partial_x +\left(
\begin{array}{cccc}
 0 & 0 & 0 & 0 \\
 0 & 0 & u^1+f^{23} & u^2+f^{24} \\
 0 & -u^1-f^{23} & 0 & -u^3+f^{34} \\
 0 & -u^2-f^{24} & u^3-f^{34} & 0 \\
\end{array}
\right),
\end{equation}
\begin{equation}
    \mathcal{A}_{4,3}=   \left(
\begin{array}{cccc}
 0 & 0 & 0 & \alpha \\
 0 & \alpha & 0 & 0 \\
 0 & 0 & \alpha & 0 \\
 \alpha & 0 & 0 & \beta \\
\end{array}
\right)\partial_x +\left(
\begin{array}{cccc}
 0 & 0 & 0 & 0 \\
 0 & 0 & u^1+f^{23} & -u^3+f^{24} \\
 0 & -u^1-f^{23} & 0 & u^2+f^{34} \\
 0 & u^3-f^{24} & -u^2-f^{34} & 0 \\
\end{array}
\right).
\end{equation}

We observe that these cases are in some sense \emph{unprecedented}, 
because they are the only ones not belonging to any of the classes presented
in Section \ref{sec3}.

We finish $n=4$ with two Lie algebras which are direct sum of $\Sl(2,\R)$ 
with $\mathfrak{n}_{1,1}$, and $\So({3},\mathbb{R})$ with $\mathfrak{n}_{1,1}$:
\begin{equation}
 \mathcal{A}_{4,4}=  \left(
\begin{array}{cccc}
 0 & 0 & \alpha & 0 \\
 0 & \frac{\alpha}{2} & 0 & 0 \\
 \alpha & 0 & 0 & 0 \\
 0 & 0 & 0 & \beta \\
\end{array}
\right)\partial_x +\left(
\begin{array}{cccc}
 0 & u^1+f^{12} & -2 u^2+f^{13} & 0 \\
 -u^1-f^{12} & 0 & u^3+f^{23} & 0 \\
 2 u^2-f^{13} & -u^3-f^{23} & 0 & 0 \\
 0 & 0 & 0 & 0 \\
\end{array}
\right),
\end{equation}
\begin{equation}
 \mathcal{A}_{4,5}=   \left(
\begin{array}{cccc}
 \alpha & 0 & 0 & 0 \\
 0 & \alpha & 0 & 0 \\
 0 & 0 & \alpha & 0 \\
 0 & 0 & 0 & \beta \\
\end{array}
\right)\partial_x +\left(
\begin{array}{cccc}
 0 & u^3+f^{12} & -u^2+f^{13} & 0 \\
 -u^3-f^{12} & 0 & u^1+f^{23} & 0 \\
 u^2-f^{13} & -u^1-f^{23} & 0 & 0 \\
 0 & 0 & 0 & 0 \\
\end{array}
\right),
\end{equation}
where $\alpha,\beta$ and $f^{ij}$ are arbitrary real constants. 
We refer to Section \ref{directsum-sub} for a general discussion of
operators arising from this class of Lie algebras.

 \subsection{Lie algebras of dimension $n=5$}  

As usual, the simplest case is given by the abelian $5$-dimensional Lie algebra $5\mathfrak{n}_{1,1}$ so that the operator is
\begin{equation}
    \mathcal{A}_{5,1}=\eta\partial_x+f
\end{equation}
where $\eta=(a^{ij}),f=(f^{ij})$ are respectively symmetric and skew-symmetric real matrices.

Other four cases are given by direct sums which are $\mathfrak{sl}(2,\mathbb{R})\oplus 2\mathfrak{n}_{1,1}$, $\mathfrak{so}(3,\mathbb{R})\oplus 2\mathfrak{n}_{1,1}$, $\mathfrak{s}_{4,6}\oplus \mathfrak{n}_{1,1}$ and $\mathfrak{s}_{4,7}\oplus \mathfrak{n}_{1,1}$:
\begin{align}
\scriptstyle\mathcal{A}_{5,2}&=
\left(
\begin{array}{c@{\hspace{3pt}}c@{\hspace{3pt}}c@{\hspace{3pt}}c@{\hspace{3pt}}c@{\hspace{3pt}}}
 0 & 0 & \alpha & 0 & 0 \\
 0 & \frac{\alpha}{2} & 0 & 0 & 0 \\
 \alpha & 0 & 0 & 0 & 0 \\
 0 & 0 & 0 & \beta & \delta \\
 0 & 0 & 0 & \delta & \gamma \\
\end{array}
\right)\partial_x +\left(
\begin{array}{c@{\hspace{3pt}}c@{\hspace{3pt}}c@{\hspace{3pt}}c@{\hspace{3pt}}c@{\hspace{3pt}}}
 0 & u^1+f^{12} & -2 u^2+f^{13} & 0 & 0 \\
 -u^1-f^{12} & 0 & u^3+f^{23} & 0 & 0 \\
 2 u^2 -f^{13} & -u^3-f^{23} & 0 & 0 & 0 \\
 0 & 0 & 0 & 0 & f^{45} \\
 0 & 0 & 0 & -f^{45} & 0 \\
\end{array}
\right),
\\
\scriptstyle \mathcal{A}_{5,3}&=
\left(
\begin{array}{c@{\hspace{3pt}}c@{\hspace{3pt}}c@{\hspace{3pt}}c@{\hspace{3pt}}c@{\hspace{3pt}}}
 \alpha & 0 & 0 & 0 & 0 \\
 0 & \alpha & 0 & 0 & 0 \\
 0 & 0 & \alpha & 0 & 0 \\
 0 & 0 & 0 & \beta & \gamma \\
 0 & 0 & 0 & \gamma & \delta \\
\end{array}
\right)\partial_x + \left(
\begin{array}{c@{\hspace{3pt}}c@{\hspace{3pt}}c@{\hspace{3pt}}c@{\hspace{3pt}}c@{\hspace{3pt}}}
 0 & u^3+f^{12} & -u^2+ f^{13} & 0 & 0 \\
 -u^3-f^{12} & 0 & u^1+f^{23} & 0 & 0 \\
 u^2 -f^{13} & -u^1-f^{23} & 0 & 0 & 0 \\
 0 & 0 & 0 & 0 &  f^{45} \\
 0 & 0 & 0 & -f^{45} & 0 \\
\end{array}
\right),
\\
\scriptstyle \mathcal{A}_{5,4}&=
\left(
\begin{array}{c@{\hspace{3pt}}c@{\hspace{3pt}}c@{\hspace{3pt}}c@{\hspace{3pt}}c@{\hspace{3pt}}}
 0 & 0 & 0 & \alpha & 0 \\
 0 & 0 & -\alpha & 0 & 0 \\
 0 & -\alpha & 0 & 0 & 0 \\
 \alpha & 0 & 0 & \beta & \gamma \\
 0 & 0 & 0 & \gamma & \delta \\
\end{array}
\right)\partial_x +\left(
\begin{array}{c@{\hspace{3pt}}c@{\hspace{3pt}}c@{\hspace{3pt}}c@{\hspace{3pt}}c@{\hspace{3pt}}}
 0 & 0 & 0 & 0 & 0 \\
 0 & 0 & u^1+f^{23} & u^2+f^{24} & 0 \\
 0 & -u^1-f^{23} & 0 & -u^3+f^{34} & 0 \\
 0 & -u^2-f^{24} & u^3-f^{34} & 0 & f^{45} \\
 0 & 0 & 0 & -f^{45} & 0 \\
\end{array}
\right),
\\
\scriptstyle \mathcal{A}_{5,5}&=
\left(
\begin{array}{c@{\hspace{3pt}}c@{\hspace{3pt}}c@{\hspace{3pt}}c@{\hspace{3pt}}c@{\hspace{3pt}}}
 0 & 0 & 0 & \alpha & 0 \\
 0 & \alpha & 0 & 0 & 0 \\
 0 & 0 & \alpha & 0 & 0 \\
\alpha & 0 & 0 & \beta & \gamma \\
 0 & 0 & 0 & \gamma & \delta \\
\end{array}
\right)\partial_x +\left(
\begin{array}{c@{\hspace{3pt}}c@{\hspace{3pt}}c@{\hspace{3pt}}c@{\hspace{3pt}}c@{\hspace{3pt}}}
 0 & 0 & 0 & 0 & 0 \\
 0 & 0 & u^1+f^{23} & -u^3+f^{24} & 0 \\
 0 & -u^1-f^{23} & 0 & u^2+f^{34} & 0 \\
 0 & u^3-f^{24} & -u^2-f^{34} & 0 & f^{45} \\
 0 & 0 & 0 & -f^{45} & 0 \\
\end{array}
\right),
\end{align}
where $\alpha,\beta,\gamma,\delta$ and  $f^{ij}$ are arbitrary real constants. 

Finally, the last case is given by $\mathfrak{n}_{5,2}$ which is a 3-step nilpotent Lie algebra {(see subsection \ref{sss:nilpot} in the last paragraph)}:
\begin{equation}
\scriptstyle \mathcal{A}_{5,6}=
    \left(
\begin{array}{c@{\hspace{3pt}}c@{\hspace{3pt}}c@{\hspace{3pt}}c@{\hspace{3pt}}c@{\hspace{3pt}}}
 0 & 0 & 0 & \alpha & 0 \\
 0 & 0 & 0 & 0 & -\alpha \\
 0 & 0 & -\alpha & 0 & 0 \\
 \alpha & 0 & 0 & \beta & \gamma \\
 0 & -\alpha & 0 & \gamma & \delta \\
\end{array}
\right)\partial_x +\left(
\begin{array}{c@{\hspace{3pt}}c@{\hspace{3pt}}c@{\hspace{3pt}}c@{\hspace{3pt}}c@{\hspace{3pt}}}
 0 & 0 & 0 & f^{14} & f^{15} \\
 0 & 0 & 0 & f^{24} & f^{14} \\
 0 & 0 & 0 & u^2+f^{34} & u^1+f^{35} \\
 -f^{14} & -f^{24} & -u^2-f^{34} & 0 & u^3+f^{45} \\
  -f^{15} & -f^{14} & -u^1-f^{35} & -u^3-f^{45} & 0 \\
\end{array}
\right),
\end{equation}
where $\alpha,\beta,\gamma,\delta$ and  $f^{ij}$ are arbitrary real constants.

\subsection{Lie algebras of dimension $n=6$.}
In this case, formulas become very cumbersome, so we choose  not to  present them in the text but rather we will give a description of the Lie algebra associated to non-homogeneous HOs. For an explicit list of the resulting operators we refer to the \Cref{appe2} of the present work. 

As for the other dimensions, the abelian Lie algebra $6\mathfrak{n}_{1,1}$ gives us the constant operator. Next, we consider the cases given by direct sums of lower dimensional Lie algebras:
\begin{itemize}
    \item  five cases given by the direct sum of 3-dimensional Lie algebras whose operator is non-degenerate, namely $\mathfrak{sl}(2,\mathbb{R})\oplus3\mathfrak{n}_{1,1}$, $\mathfrak{so}(3,\mathbb{R})\oplus3\mathfrak{n}_{1,1}$, $\mathfrak{so}(2,2,\mathbb{R})\cong\mathfrak{sl}(2,\mathbb{R})\oplus\mathfrak{sl}(2,\mathbb{R})$, $\mathfrak{so}(4,\mathbb{R})\cong\mathfrak{so}(3,\mathbb{R})\oplus\mathfrak{so}(3,\mathbb{R})$ and $\mathfrak{sl}(2,\mathbb{R})\oplus\mathfrak{so}(3,\mathbb{R})$;
    \item three non-degenerate operators from the direct sum between the solvable Lie algebras of dimension 4 and the 2-dimensional abelian Lie algebra. Respectively, we have $\mathfrak{s}_{4,6}\oplus 2\mathfrak{n}_{1,1}$, $\mathfrak{s}_{4,7}\oplus 2\mathfrak{n}_{1,1}$;
    \item the unique nilpotent Lie algebra of dimension 5 and the 1-dimensional abelian Lie algebra: $\mathfrak{n}_{5,2}\oplus\mathfrak{n}_{1,1}$;
\end{itemize}
Furthermore, we can obtain  the first case of a 2-step nilpotent Lie algebra $\mathfrak{n}_{6,1}$ with a non-degenerate scalar product. Looking at the classification of the Winternitz-\v{S}nobl book \cite{snowin}, we have 6 cases of solvable Lie algebras $\mathfrak{s}_{6,k}$, with $k=162,\dots,167$, whose operator is non-degenerate. Finally, the last two cases are represented by the simple Lie algebra $\mathfrak{so}(1,3,\mathbb{R})$ and the Lie algebra $\mathfrak{sl}(3,\mathbb{R})\ltimes 3\mathfrak{n}_{1,1}$, which is Levi decomposable.

\section{Conclusions}
In this paper, we investigated non-homogeneous hydrodynamic type operators which are composed by a first-order homogeneous one (also named after Dubrovin and Novikov) and a compatible Poisson tensor. Here, we studied operators in Darboux form (i.e. with constant  leading coefficient entries) for which the Poisson tensor is linear in the field variables. We studied their geometric interpretation in details, showed that the leading coefficient is related to the Casimir of Lie algebra structure associated to the Poisson tensor and listed ways to construct non-homogeneous Hamiltonian operators for classes of Lie algebras. As a result, we are able to characterise up to dimension $n=6$ the operators whose leading coefficient is non-degenerate. 

As a key example, we presented the Korteweg-de Vries equation in an equivalent form, i.e. written as a non-homogeneous quasilinear system. In this case, the bi-Hamiltonian property was proved by Mokhov by showing the existence of two compatible operators of type $1+0$. We stress that the first one is  geometrically described in the framework we introduced in this paper.  However, the latter example reveals that further investigations are needed to study from  the point of view of the Lie algebras the bi-Hamiltonian property. 

As a preliminary result, let us consider two non-homogeneous operators in Darboux form
\begin{equation}
    \mathcal{A}^{ij}=g_1^{ij}\partial_x+\left(c^{ij}_{1,k}u^k+f^{ij}_1\right),\qquad \mathcal{B}^{ij}=g^{ij}_2\partial_x+\left(c^{ij}_{2,k}u^k+f^{ij}_2\right).
\end{equation}

Albeit this choice can seem restrictive, the example of the KdV equation presented in Section \ref{sec4} (see operators in \eqref{kdvvv}) is exactly of this type, allowing the possibility that the second structure is degenerate.  

 We now recall that $\mathcal{A}$ and $\mathcal{B}$ are compatible by definition if the pencil $\mathcal{A}+\lambda\mathcal{B}$ preserves the Hamiltonian property. In our case, we require that 
$$\left(g^{ij}_1+\lambda g^{ij}_2\right)\partial_x+\left(c^{ij}_{1,k}u^k+f^{ij}_1+\lambda(c^{ij}_{2,k}u^k+f^{ij}_2)\right)$$
is Hamiltonian. Note that $g^{ij}_1+\lambda g^{ij}_2$ is a constant scalar product, hence the first-order part is always Hamiltonian. Then, we only need to study when the tensor
\begin{equation}
   \omega^{ij}_\lambda= c^{ij}_{1,k}u^k+f^{ij}_1+\lambda(c^{ij}_{2,k}u^k+f^{ij}_2)
\end{equation} is a Poisson tensor, i.e. we require $\tilde{c}^{ij}_k:=c^{ij}_{1,k}+\lambda c^{ij}_{2,k}$ to satisfy the Jabobi identity and $f^{ij}_1+\lambda f^{ij}_2$ to be a  2-cocycle for the Lie algebra whose structure constants are given by $\tilde{c}^{ij}_k$. 
As a consequence, the following result holds
\begin{theorem}
    Two non-homogeneous $1+0$ Hamiltonian operators in Darboux form are compatible if and only if 
   \begin{subequations}\label{eq:DRV_condition} \begin{align}
        &c_{2,p}^{ij}c_{1,s}^{pk}+c_{2,p}^{jk}c_{1,s}^{pi}+c_{2,p}^{ki}c_{1,s}^{pj} + c_{1,p}^{ij}c_{2,s}^{pk}+c_{1,p}^{jk}c_{2,s}^{pi}+c_{1,p}^{ki}c_{2,s}^{pj} = 0,\label{drv1}\\
 &c_{2,p}^{ij}f_1^{pk}+c_{2,p}^{jk}f_1^{pi}+c_{2,p}^{ki}f_1^{pj} +c_{1,p}^{ij}f_2^{pk}+c_{1,p}^{jk}f_2^{pi}+c_{1,p}^{ki}f_2^{pj} = 0,\label{drv2}
    \\
    \label{condgc2}
&g_1^{is}c_{2,s}^{jk}+g_1^{js}c_{2,s}^{ik}+{g}_2^{is}c_{1,s}^{jk}+{g}_2^{js}c_{1,s}^{ik} = 0. 
    \end{align}
    \end{subequations}
\end{theorem}
\begin{proof}
The proof easily follows collecting for $\lambda^k$ ($k=0,1,2$) the conditions required for the Hamiltonianity of $\mathcal{A}+\lambda\mathcal{B}$. From the Jacobi identity for $\tilde{c}^{ij}_k$ we obtain \eqref{drv1}, from the 2-cocicle requirement condition \eqref{drv2} is needed and finally from the compatibility between the scalar product $g^{ij}_1+\lambda g^{ij}_2$ and the Poisson tensor,  \eqref{condgc2} is obtained.
\end{proof}

Note that condition \eqref{drv1} is equivalent to the requirment that the two Lie algebra structures are compatible, i.e. that $[\;,\;]_1+\lambda [\;,\;]_2$ is again a Lie bracket (see \cite{liecom1,liecom2,liecom3}). 

An application of the previous result is shown in the following Example:

\begin{example}
    Let $\{w^1,w^2,w^3\}$ be a basis of $\mathfrak{g}$ such that
\[
[w^1,w^2] = -2w^3,\quad [w^1,w^3]=2w^2,\quad [w^2,w^3]=2w^1,
\]
and let $\{\tilde{w}^1,\tilde{w}^2,\tilde{w}^3\}$ be a basis of $\tilde{\mathfrak{g}}$ such that
\[
[\tilde{w}^1,\tilde{w}^2] = \tilde{w}^1 -\tilde{w}^3,\quad [\tilde{w}^2,\tilde{w}^3] = \tilde{w}^3 -\tilde{w}^1.
\]
Using the condition of compatibility of the scalar product, 
we obtain 
\[
\eta = \begin{pmatrix}
-\alpha & 0 & 0 \\
0 & \alpha & 0 \\
0 & 0 & \alpha
\end{pmatrix},\qquad 
\tilde{\eta} = \begin{pmatrix}
\tilde{\alpha} & \tilde{\beta} & \tilde{\alpha} \\
\tilde{\beta} & \tilde{\gamma} & \tilde{\beta} \\
\tilde{\alpha} & \tilde{\beta} & \tilde{\alpha}
\end{pmatrix}.
\]
After setting the compatibility \eqref{condgc2}
where $c_{1,s}^{ij}$ and $c_{2,s}^{ij}$ are the structure constants of $\mathfrak{g}$ and $\tilde{\mathfrak{g}}$ respectively, one has
\[
\eta = \begin{pmatrix}
-\alpha & 0 & 0 \\
0 & \alpha & 0 \\
0 & 0 & \alpha
\end{pmatrix},
\qquad
\tilde{\eta} = \begin{pmatrix}
-\frac{\alpha}{2} & 0 & -\frac{\alpha}{2} \\
0 & 0 & 0 \\
-\frac{\alpha}{2} & 0 & -\frac{\alpha}{2}
\end{pmatrix}.
\]

Note that the condition (\ref{drv1}) is identically zero, so it is satisfied. The same happens for condition (\ref{drv2}), so we have no constraints on cocycles. Then, the operators for the algebras $\mathfrak{g}$ and $\tilde{\mathfrak{g}}$ are
\begin{equation}
\begin{pmatrix}
-\alpha & 0 & 0 \\
0 & \alpha & 0 \\
0 & 0 & \alpha
\end{pmatrix}\partial_x + 
\begin{pmatrix}
0 & -2w^3 & 2w^2\\
2w^3 & 0 & 2w^1\\
-2w^2 & -2w^1 & 0
\end{pmatrix} +
\begin{pmatrix}
0 & f_1^{12} & f_1^{13}\\
-f_1^{12} & 0 & f_1^{23}\\
-f_1^{13} & - f_1^{23} & 0
\end{pmatrix}
\end{equation}
and
\begin{equation}
\begin{pmatrix}
-\frac{\alpha}{2} & 0 & -\frac{\alpha}{2} \\
0 & 0 & 0 \\
-\frac{\alpha}{2} & 0 & -\frac{\alpha}{2}
\end{pmatrix}\partial_x + 
(w^1-w^3)\begin{pmatrix}
0 & 1 & 0\\
-1 & 0 & -1\\
0 & 1 & 0
\end{pmatrix} +
\begin{pmatrix}
0 & f_2^{12} & f_2^{13}\\
-f_2^{12} & 0 & f_2^{23}\\
-f_2^{13} & - f_2^{23} & 0
\end{pmatrix},
\end{equation}
respectively.

The Lie algebra $\mathfrak{g}$ is isomorphic to Lie algebra $\mathfrak{sl}(2,\mathbb{R})$, whose non--zero Lie bracket are
\[
[w^1,w^3] = -2w^2,\quad [w^1,w^2]=w^1,\quad [w^2,w^3]=w^3,
\]
and $\tilde{\mathfrak{g}}$ is isomorphic to  $\mathfrak{n}_{3,1}$.



\end{example}

As it turned out in the Example, a crucial role is here played by the second Hamiltonian structure, which often happens to be degenerate.  This is a common fact when dealing with non-homogeneous HOs of hydrodynamic type as deeply investigated in \cite{DellAVer1}. 

We plan to investigate the  bi-Hamiltonian structures of this case into details in another future work on degenerate operators and compatible pairs in low dimensions. We observe despite being Lie algebras classified, the previous example shows that we can assume only one of the two structures to be in canonical form in the sense of Lie algebras \cite[Part 3]{snowin}. Indeed, we can choose the second structure to be compatible with the first one up to Lie algebra isomorphisms. For instance, in the KdV example we cannot choose as second structure exactly in the canonical form of $\mathfrak{n}_{3,1}$ but we have to properly select the representative.

\subsection*{Acknowledgments}
The authors thank the interesting comments and remarks of Michele Graffeo. 

All the authors are member of Gruppo Nazionale per la Fisica Matematica (GNFM) of the Istituto Nazionale di Alta Matematica (INdAM). GG's and PV’s research was partially supported by the research project Mathematical Methods in Non-Linear Physics (MMNLP) and by the Commissione Scientifica Nazionale – Gruppo 4 – Fisica Teorica of the Istituto Nazionale di Fisica Nucleare (INFN). PV is also supported by the project “An artificial intelligence approach for risk assessment and prevention of low back pain: towards precision spine care”, PNRR-MAD-2022-12376692, CUP: J43C22001510001 funded by the European Union - Next Generation EU - NRRP M6C2 - Investment 2.1 Enhancement and strengthening of biomedical research in the NHS.

\paragraph{Data Availability} Not applicable. \paragraph{Declarations}

 The authors have no competing interests to declare that are relevant to the content of this article.

\clearpage 

\appendix

\section{Hamiltonian operator associated to the split form of the $\galg_2$ Lie algebra.}\label{app:g2}

In this Appendix we show the 1+0 Hamiltonian operator associated to the exceptional
Lie algebra $\galg_2$ in its split form. The Lie algebra $\galg_2$ is the exceptional Lie algebra
of smallest dimension being $\dim\galg_2$. For more information and a proof of the construction
of this exceptional Lie algebra we refer to \cite{FultonHarris1991}, while a survey of
its history and importance in geometric problem is given~\cite{Agricola2008}.

The split form of the Lie algebra $\galg_2=\Span\{e^1,\ldots,e^{14}\}$ has the following commutation relations:
\begin{align*}
    \begin{split}
      &  [e^1, e^3] = 2e^3, [e^1, e^4] = -3e^4, [e^1, e^5] = -e^5,  [e^1, e^6] = e^6, [e^1, e^7] = 3e^7, \\
      &
        [e^1, e^9] = -2e^9, [e^1, e^{10}] = 3e^{10}, [e^1, e^{11}] = e^{11}, [e^1, e^{12}] = -e^{12}, [e^1, e^{13}] = -3e^{13}, 
        \\
       & [e^2, e^3] = -e^3, [e^2, e^4] = 2e^4, [e^2, e^5] = e^5, [e^2, e^7] = -e^7, [e^2, e^8] = e^8, [e^2, e^9] = e^9, \\ 
       &[e^2, e^{10}] = -2e^{10}, [e^2, e^{11}] = -e^{11},  [e^2, e^{13}] = e^{13}, [e^2, e^{14}] = -e^{14}, [e^3, e^4] = e^5,\\
       & [e^3, e^5] = 2e^6, [e^3, e^6] = -3e^7, [e^3, e^9] = -e^1, 
        [e^3, e^{11}] = -3e^{10}, [e^3, e^{12}] = -2e^{11},\\
       & [e^3, e^{13}] = e^{12}, [e^4, e^7] = -e^8, [e^4, e^{10}] = -e^2, 
          [e^4, e^{11}] = e^9, [e^4, e^{14}] = e^{13},\\
       & [e^5, e^6] = -3e^8, [e^5, e^9] = 3e^4, [e^5, e^{10}] = -e^3, 
         [e^5, e^{11}] = -3e^2-e^1, [e^5, e^{12}] = 2e^{9}, \\
       & [e^5, e^{14}] = e^{12}, [e^6, e^9] = 2e^5, [e^6, e^{11}] = -2e^3, 
         [e^6, e^{12}] = -3e^2-2e^1, [e^6, ^{13}] = -e^9,\\
       & [e^6, e^{14}] = -e^{11}, [e^7, e^9] = -e^6, [e^7, e^{12}] = e^3, 
         [e^7, e^{13}] = -e^2-e^1, [e^7, e^{14}] = -e^{10},\\
       & [e^8, e^{10}] = -e^7, [e^8, e^{11}] = -e^6, [e^8, e^{12}] = e^5, 
         [e^8, e^{13}] = e^4, [e^8, e^{14}] = -2e^2-e^1,\\
       & [e^9, e^{10}] = e^{11}, [e^9, e^{11}] = 2e^{12}, [e^9, e^{12}] = -3e^{13},[e^{10}, e^{13}] = -e^{14}, [e^{11}, e^{12}] = -3e^{14}
    \end{split}
\end{align*}
Within this choice, the Cartan subalgebra of $\galg_2$ is generated by $e^1$ and $e^2$ ($\galg_2$ has rank 2).

\begin{equation*}
\eta = 
\left(
\begin{array}{cccccccccccccc}
 -6 \alpha & 3 \alpha & 0 & 0 & 0 & 0 & 0 & 0 & 0 & 0 & 0 & 0 & 0 & 0 \\
 3 \alpha & -2 \alpha & 0 & 0 & 0 & 0 & 0 & 0 & 0 & 0 & 0 & 0 & 0 & 0 \\
 0 & 0 & 0 & 0 & 0 & 0 & 0 & 0 & 3 \alpha & 0 & 0 & 0 & 0 & 0 \\
 0 & 0 & 0 & 0 & 0 & 0 & 0 & 0 & 0 & \alpha & 0 & 0 & 0 & 0 \\
 0 & 0 & 0 & 0 & 0 & 0 & 0 & 0 & 0 & 0 & 3 \alpha & 0 & 0 & 0 \\
 0 & 0 & 0 & 0 & 0 & 0 & 0 & 0 & 0 & 0 & 0 & 3 \alpha & 0 & 0 \\
 0 & 0 & 0 & 0 & 0 & 0 & 0 & 0 & 0 & 0 & 0 & 0 & \alpha & 0 \\
 0 & 0 & 0 & 0 & 0 & 0 & 0 & 0 & 0 & 0 & 0 & 0 & 0 & \alpha \\
 0 & 0 & 3 \alpha & 0 & 0 & 0 & 0 & 0 & 0 & 0 & 0 & 0 & 0 & 0 \\
 0 & 0 & 0 & \alpha & 0 & 0 & 0 & 0 & 0 & 0 & 0 & 0 & 0 & 0 \\
 0 & 0 & 0 & 0 & 3 \alpha & 0 & 0 & 0 & 0 & 0 & 0 & 0 & 0 & 0 \\
 0 & 0 & 0 & 0 & 0 & 3 \alpha & 0 & 0 & 0 & 0 & 0 & 0 & 0 & 0 \\
 0 & 0 & 0 & 0 & 0 & 0 & \alpha & 0 & 0 & 0 & 0 & 0 & 0 & 0 \\
 0 & 0 & 0 & 0 & 0 & 0 & 0 & \alpha & 0 & 0 & 0 & 0 & 0 & 0 \\
\end{array}
\right)
\end{equation*}
\begin{landscape}
{\begin{equation*}
\omega = 
\left(
\begin{array}{cccccccccccccc}
 0 & 0 & 2 u^3 & -3 u^4 & -u^5 & u^6 & 3 u^7 & 0 & -2 u^9 & 3 u^{10} & u^{11} & -u^{12} & -3 u^{13} & 0 \\
 0 & 0 & -u^3 & 2 u^4 & u^5 & 0 & -u^7 & u^8 & u^9 & -2 u^{10} & -u^{11} & 0 & u^{13} & -u^{14} \\
 -2 u^3 & u^3 & 0 & u^5 & 2 u^6 & -3 u^7 & 0 & 0 & -u^1 & 0 & -3 u^{10} & -2 u^{11} & u^{12} & 0 \\
 3 u^4 & -2 u^4 & -u^5 & 0 & 0 & 0 & -u^8 & 0 & 0 & -u^2 & u^9 & 0 & 0 & u^{13} \\
 u^5 & -u^5 & -2 u^6 & 0 & 0 & -3 u^8 & 0 & 0 & 3 u^4 & -u^3 & -u^1-3 u^2 & 2 u^9 & 0 & u^{12} \\
 -u^6 & 0 & 3 u^7 & 0 & 3 u^8 & 0 & 0 & 0 & 2 u^5 & 0 & -2 u^3 & -2 u^1-3 u^2 & -u^9 & -u^{11} \\
 -3 u^7 & u^7 & 0 & u^8 & 0 & 0 & 0 & 0 & -u^6 & 0 & 0 & u^3 & -u^1-u^2 & -u^{10} \\
 0 & -u^8 & 0 & 0 & 0 & 0 & 0 & 0 & 0 & -u^7 & -u^6 & u^5 & u^4 & -u^1-2 u^2 \\
 2 u^9 & -u^9 & u^1 & 0 & -3 u^4 & -2 u^5 & u^6 & 0 & 0 & u^{11} & 2 u^{12} & -3 u^{13} & 0 & 0 \\
 -3 u^{10} & 2 u^{10} & 0 & u^2 & u^3 & 0 & 0 & u^7 & -u^{11} & 0 & 0 & 0 & -u^{14} & 0 \\
 -u^{11} & u^{11} & 3 u^{10} & -u^9 & u^1+3 u^2 & 2 u^3 & 0 & u^6 & -2 u^{12} & 0 & 0 & -3 u^{14} & 0 & 0 \\
 u^{12} & 0 & 2 u^{11} & 0 & -2 u^9 & 2 u^1+3 u^2 & -u^3 & -u^5 & 3 u^{13} & 0 & 3 u^{14} & 0 & 0 & 0 \\
 3 u^{13} & -u^{13} & -u^{12} & 0 & 0 & u^9 & u^1+u^2 & -u^4 & 0 & u^{14} & 0 & 0 & 0 & 0 \\
 0 & u^{14} & 0 & -u^{13} & -u^{12} & u^{11} & u^{10} & u^1+2 u^2 & 0 & 0 & 0 & 0 & 0 & 0 \\
\end{array}
\right)
\end{equation*}
}
{
\begin{equation*}
f = 
\left(
\begin{array}{cccccccccccccc}
 0 & 0 & a^1 & a^2 & a^3 & a^4 & a^5 & 0 & 0 & -3 a^7 & 0 & 0 & -3 a^7 & 0 \\
 0 & 0 & -\frac{a^1}{2} & -\frac{2a^2}{3} & -a^3 & 0 & -\frac{a^5}{3} & a^6 & 0 & 2 a^7 & 0 & 0 & a^7 & a^{10} \\
 -a^1 & \frac{a^1}{2} & 0 & -a^3 & 2 a^4 & -a^5 & 0 & 0 & a^9-a^8 & 0 & 3 a^7 & 0 & 0 & 0 \\
 -a^2 & \frac{2 a^2}{3} & a^3 & 0 & 0 & 0 & -a^6 & 0 & 0 & a^8 & 0 & 0 & 0 & a^7 \\
 -a^3 & a^3 & -2 a^4 & 0 & 0 & -3 a^6 & 0 & 0 & -a^2 & -\frac{a^1}{2} & 2 a^8+a^9 & 0 & 0 & 0 \\
 -a^4 & 0 & a^5 & 0 & 3 a^6 & 0 & 0 & 0 & -2 a^3 & 0 & -a^1 & a^8+2 a^9 & 0 & 0 \\
 -a^5 & \frac{a^5}{3} & 0 & a^6 & 0 & 0 & 0 & 0 & -a^4 & 0 & 0 & \frac{a^1}{2} & a^9 & a^7 \\
 0 & -a^6 & 0 & 0 & 0 & 0 & 0 & 0 & 0 & -\frac{a^5}{3} & -a^4 & -a^3 & -\frac{a^2}{3} & a^8+a^9 \\
 0 & 0 & a^8-a^9 & 0 & a^2 & 2 a^3 & a^4 & 0 & 0 & 0 & 0 & -3 a^7 & 0 & 0 \\
 3 a^7 & -2 a^7 & 0 & -a^8 & \frac{a^1}{2} & 0 & 0 & \frac{a^5}{3} & 0 & 0 & 0 & 0 & a^{10} & 0 \\
 0 & 0 & -3 a^7 & 0 & -2 a^8-a^9 & a^1 & 0 & a^4 & 0 & 0 & 0 & 3 a^{10} & 0 & 0 \\
 0 & 0 & 0 & 0 & 0 & -a^8-2 a^9 & -\frac{a^1}{2} & a^3 & 3 a^7 & 0 & -3 a^{10} & 0 & 0 & 0 \\
 3 a^7 & -a^7 & 0 & 0 & 0 & 0 & -a^9 & \frac{a^2}{3} & 0 & -a^{10} & 0 & 0 & 0 & 0 \\
 0 & -a^{10} & 0 & -a^7 & 0 & 0 & -a^7 & -a^8-a^9 & 0 & 0 & 0 & 0 & 0 & 0 \\
\end{array}
\right)
\end{equation*}}
\end{landscape}

\newpage

\section{Explicit form of operators in 6 components}\label{appe2}

Also in this case, the abelian Lie algebra $6\mathfrak{n}_{1,1}$ gives us the constant operator
\begin{equation*}
\scriptstyle
  \mathcal{A}_{6,1}=  \eta\partial_x +f,
\end{equation*}
where $\eta=(a^{ij}),f=(f^{ij})$ are respectively symmetric and skew-symmetric real constants.

Additionally, we have 5 cases given by the direct sum of 3-dimensional Lie algebras whose operator is non-degenerate:

\begin{itemize}
    \item $\mathfrak{sl}(2,\mathbb{R})\oplus3\mathfrak{n}_{1,1}$ with associated
        operator:
        \begin{equation*}
            \begin{aligned}
            \mathcal{A}_{6,2} &=
            \begin{pmatrix}
            0 & 0 & g_{13} & 0 & 0 & 0 \\
            0 & 2 g_{13} & 0 & 0 & 0 & 0 \\
            g_{13} & 0 & 0 & 0 & 0 & 0 \\
            0 & 0 & 0 & g_{44} & g_{45} & g_{46} \\
            0 & 0 & 0 & g_{45} & g_{55} & g_{56} \\
            0 & 0 & 0 & g_{46} & g_{56} & g_{66}
            \end{pmatrix}
            \partial_x
            \\
            &+\begin{pmatrix}
                0 & -2 u^1+f^{12} & u^2+f^{13} & 0 & 0 & 0 \\
            2 u^1-f^{12} & 0 & -2 u^3+f^{23} & 0 & 0 & 0 \\
            -u^2-f^{13} & 2 u^3-f^{23} & 0 & 0 & 0 & 0 \\
            0 & 0 & 0 & 0 & f^{45} & f^{46} \\
            0 & 0 & 0 & -f^{45} & 0 & f^{56} \\
            0 & 0 & 0 & -f^{46} & -f^{56} & 0
            \end{pmatrix},    
            \end{aligned}
        \end{equation*}        
    \item $\mathfrak{so}(3,\mathbb{R})\oplus3\mathfrak{n}_{1,1}$ with associated
        operator:
        \begin{equation*}
            \begin{aligned}
            \mathcal{A}_{6,3} &=
            \begin{pmatrix}
                g_{22} & 0 & 0 & 0 & 0 & 0 \\
                0 & 2 g_{22} & 0 & 0 & 0 & 0 \\
                0 & 0 & g_{22} & 0 & 0 & 0 \\
                0 & 0 & 0 & g_{44} & g_{45} & g_{46} \\
                0 & 0 & 0 & g_{45} & g_{55} & g_{56} \\
                0 & 0 & 0 & g_{46} & g_{56} & g_{66} \\
            \end{pmatrix}
            \partial_x 
            \\
            &+ 
            \begin{pmatrix}
                0 & -u^3+f^{12} & u^2+f^{13} & 0 & 0 & 0 \\
                u^3-f^{12} & 0 & -u^1+f^{23} & 0 & 0 & 0 \\
                -u^2-f^{13} & u^1-f^{23} & 0 & 0 & 0 & 0 \\
                0 & 0 & 0 & 0 & f^{45} & f^{46} \\
                0 & 0 & 0 & -f^{45} & 0 & f^{56} \\
                0 & 0 & 0 & -f^{46} & -f^{56} & 0
            \end{pmatrix},    
            \end{aligned}
        \end{equation*}
    \item $\mathfrak{so}(2,2,\mathbb{R})\cong
        \mathfrak{sl}(2,\mathbb{R})\oplus\mathfrak{sl}(2,\mathbb{R})$ with 
        associated operator:
        \begin{align*}
            \mathcal{A}_{6,4} &=
                \begin{pmatrix}
                    0 & 0 & g_{13} & 0 & 0 & 0 \\
                    0 & 2 g_{13} & 0 & 0 & 0 & 0 \\
                    g_{13} & 0 & 0 & 0 & 0 & 0 \\
                    0 & 0 & 0 & 0 & 0 & g_{46} \\
                    0 & 0 & 0 & 0 & 2 g_{46} & 0 \\
                    0 & 0 & 0 & g_{46} & 0 & 0
                \end{pmatrix}
                \partial_x\notag
                \\
                &+ 
                \begin{pmatrix}
                    0 & -2 u^1 & u^2 & 0 & 0 & 0 \\
                    2 u^1 & 0 & -2 u^3 & 0 & 0 & 0 \\
                    -u^2 & 2 u^3 & 0 & 0 & 0 & 0 \\
                    0 & 0 & 0 & 0 & -2 u^4 & u^5 \\
                    0 & 0 & 0 & 2 u^4 & 0 & -2 u^6 \\
                    0 & 0 & 0 & -u^5 & 2 u^6 & 0
            \end{pmatrix}
            \\
                &+ 
                \begin{pmatrix}
                    0 & f^{12} & f^{13} & 0 & 0 & 0 \\
                    -f^{12} & 0 & f^{23} & 0 & 0 & 0 \\
                    -f^{13} & 2 -f^{23} & 0 & 0 & 0 & 0 \\
                    0 & 0 & 0 & 0 & f^{45} & f^{46} \\
                    0 & 0 & 0 & 2 -f^{45} & 0 & f^{56} \\
                    0 & 0 & 0 & -f^{46} & -f^{56} & 0
            \end{pmatrix},\notag
        \end{align*}
    \item $\mathfrak{so}(4,\mathbb{R}) \cong 
        \mathfrak{so}(3,\mathbb{R})\oplus\mathfrak{so}(3,\mathbb{R})$ with 
        associated operator:
        \begin{align*}
                \mathcal{A}_{6,5} &=
                \begin{pmatrix}
                    g_{22} & 0 & 0 & 0 & 0 & 0 \\
                    0 & g_{22} & 0 & 0 & 0 & 0 \\
                    0 & 0 & g_{22} & 0 & 0 & 0 \\
                    0 & 0 & 0 & g_{55} & 0 & 0 \\
                    0 & 0 & 0 & 0 & g_{55} & 0 \\
                    0 & 0 & 0 & 0 & 0 & g_{55}
                \end{pmatrix}
                \partial_x\notag
                \\
                &+
                \begin{pmatrix}
                    0 & -u^3 & u^2 & 0 & 0 & 0 \\
                    u^3 & 0 & -u^1 & 0 & 0 & 0 \\
                    -u^2 & u^1 & 0 & 0 & 0 & 0 \\
                    0 & 0 & 0 & 0 & -u^6 & u^5 \\
                    0 & 0 & 0 & u^6 & 0 & -u^4 \\
                    0 & 0 & 0 & -u^5 & u^4 & 0
                \end{pmatrix}
                \\
                &+
                \begin{pmatrix}
                    0 & f^{12} & f^{13} & 0 & 0 & 0 \\
                    -f^{12} & 0 & f^{23} & 0 & 0 & 0 \\
                    -f^{13} & -f^{23} & 0 & 0 & 0 & 0 \\
                    0 & 0 & 0 & 0 & f^{45} & f^{46} \\
                    0 & 0 & 0 & -f^{45} & 0 & f^{56} \\
                    0 & 0 & 0 & -f^{46} & -f^{56} & 0
                \end{pmatrix},\notag
            \end{align*}
    \item $\mathfrak{sl}(2,\mathbb{R})\oplus\mathfrak{so}(3,\mathbb{R})$ with 
        associated operator:
        \begin{align*}
                \mathcal{A}_{6,6} &=
                \begin{pmatrix}
 0 & 0 & g_{13} & 0 & 0 & 0 \\
 0 & 2 g_{13} & 0 & 0 & 0 & 0 \\
 g_{13} & 0 & 0 & 0 & 0 & 0 \\
 0 & 0 & 0 & g_{55} & 0 & 0 \\
 0 & 0 & 0 & 0 & g_{55} & 0 \\
 0 & 0 & 0 & 0 & 0 & g_{55} \\
\end{pmatrix}
\partial_x\notag
            \\
            &+ 
            \begin{pmatrix}
                0 & -2 u^1 & u^2 & 0 & 0 & 0 \\
                2 u^1 & 0 & -2 u^3 & 0 & 0 & 0 \\
                -u^2 & 2 u^3 & 0 & 0 & 0 & 0 \\
                0 & 0 & 0 & 0 & -u^6 & u^5 \\
                0 & 0 & 0 & u^6 & 0 & -u^4 \\
                0 & 0 & 0 & -u^5 & u^4 & 0
            \end{pmatrix}
            \\
            &+ 
            \begin{pmatrix}
                0 & f^{12} & f^{13} & 0 & 0 & 0 \\
                -f^{12} & 0 & f^{23} & 0 & 0 & 0 \\
                -f^{13} & -f^{23} & 0 & 0 & 0 & 0 \\
                0 & 0 & 0 & 0 & f^{45} & f^{46} \\
                0 & 0 & 0 & -f^{45} & 0 & f^{56} \\
                0 & 0 & 0 & -f^{46} & -f^{56} & 0
            \end{pmatrix}.\notag
            \end{align*}   
\end{itemize}

Furthermore, we can obtain three non-degenerate operators from 
the direct sum between the solvable Lie algebras of dimension 4 
and the two-dimensional abelian Lie algebra and the unique nilpotent 
Lie algebra of dimension 5 and the one-dimensional abelian Lie algebra:

\begin{itemize}
    \item $\mathfrak{s}_{4,6}\oplus 2\mathfrak{n}_{1,1}$ with 
        associated operator:
    \begin{equation*}
        \begin{aligned}
            \mathcal{A}_{6,7}&=
            \begin{pmatrix}
                0 & 0 & 0 & g_{14} & 0 & 0 \\
                0 & 0 & g_{14} & 0 & 0 & 0 \\
                0 & g_{14} & 0 & 0 & 0 & 0 \\
                g_{14} & 0 & 0 & g_{44} & g_{45} & g_{46} \\
                0 & 0 & 0 & g_{45} & g_{55} & g_{56} \\
                0 & 0 & 0 & g_{46} & g_{56} & g_{66}
            \end{pmatrix}
            \partial_x 
            \\
            &+
            \begin{pmatrix}
                0 & 0 & 0 & 0 & 0 & 0 \\
                0 & 0 & -u^1+f^{23} & u^2+f^{24} & 0 & 0 \\
                0 & u^1-f^{23} & 0 & -u^3+f^{34} & 0 & 0 \\
                0 & -u^2-f^{24} & u^3-f^{34} & 0 & f^{45} & f^{46} \\
                0 & 0 & 0 & -f^{45} & 0 & f^{56} \\
                0 & 0 & 0 & -f^{46} & -f^{56} & 0 
            \end{pmatrix}
        \end{aligned}
        \end{equation*}
    \item $\mathfrak{s}_{4,7}\oplus 2\mathfrak{n}_{1,1}$ with 
        associated operator:
        \begin{equation*}
            \begin{aligned}
            \mathcal{A}_{6,8} &=
            \begin{pmatrix}
                0 & 0 & 0 & g_{14} & 0 & 0 \\
                0 & -g_{14} & 0 & 0 & 0 & 0 \\
                0 & 0 & -g_{14} & 0 & 0 & 0 \\
                g_{14} & 0 & 0 & g_{44} & g_{45} & g_{46} \\
                0 & 0 & 0 & g_{45} & g_{55} & g_{56} \\
                0 & 0 & 0 & g_{46} & g_{56} & g_{66}
            \end{pmatrix}
            \partial_x 
            \\
            &+
            \begin{pmatrix}
                0 & 0 & 0 & 0 & 0 & 0 \\
                0 & 0 & -u^1+f^{23} & -u^3+f^{24} & 0 & 0 \\
                0 & u^1-f^{23} & 0 & u^2+f^{34} & 0 & 0 \\
                0 & u^3-f^{24} & -u^2-f^{34} & 0 & f^{45} & f^{46} \\
                0 & 0 & 0 & -f^{45} & 0 & f^{56} \\
                0 & 0 & 0 & -f^{46} & -f^{56} & 0 \\
            \end{pmatrix},    
            \end{aligned}    
\end{equation*}
    \item $\mathfrak{n}_{5,2}\oplus\mathfrak{n}_{1,1}$ with 
        associated operator:
        \begin{equation*}
            \begin{aligned}
                \mathcal{A}_{6,9} &=
                \begin{pmatrix}
                    0 & 0 & 0 & g_{14} & 0 & 0 \\
                    0 & 0 & 0 & 0 & -g_{14} & 0 \\
                    0 & 0 & -g_{14} & 0 & 0 & 0 \\
                    g_{14} & 0 & 0 & g_{44} & g_{45} & g_{46} \\
                    0 & -g_{14} & 0 & g_{45} & g_{55} & g_{56} \\
                    0 & 0 & 0 & g_{46} & g_{56} & g_{66}
                \end{pmatrix}\partial_x 
                \\
                &+
                \begin{pmatrix}
                    0 & 0 & 0 & f^{14} & f^{15} & 0 \\
                    0 & 0 & 0 & f^{24} & f^{14} & 0 \\
                    0 & 0 & 0 & -u^2+f^{34} & -u^1+f^{35} & 0 \\
                    -f^{14} & -f^{24} & u^2-f^{34} & 0 & -u^3+f^{45} & f^{46} \\
                    -f^{15} & -f^{14} & u^1-f^{35} & u^3-f^{45} & 0 & f^{56} \\
                    0 & 0 & 0 & -f^{46} & -f^{56} & 0 \\
            \end{pmatrix}.
            \end{aligned}
        \end{equation*}
\end{itemize}

In dimension 6 we have the first case of a 2-step nilpotent Lie algebra 
$\mathfrak{n}_{6,1}$ with a non-degenerate scalar product, as proven in
\ref{thm:twostepnilpotent}. The structure of the operator is the following:
\begin{equation*}
    \begin{aligned}
        \mathcal{A}_{6,10} &=
        \begin{pmatrix}
            0 & 0 & 0 & g_{14} & 0 & 0 \\
            0 & 0 & 0 & 0 & -g_{14} & 0 \\
            0 & 0 & -g_{14} & 0 & 0 & 0 \\
            g_{14} & 0 & 0 & g_{44} & g_{45} & g_{46} \\
            0 & -g_{14} & 0 & g_{45} & g_{55} & g_{56} \\
            0 & 0 & 0 & g_{46} & g_{56} & g_{66}
        \end{pmatrix}
        \partial_x 
        \\
        &+
        \begin{pmatrix}
            0 & 0 & 0 & f^{14} & f^{15} & 0 \\
            0 & 0 & 0 & f^{24} & f^{14} & 0 \\
            0 & 0 & 0 & -u^2+f^{34} & -u^1+f^{35} & 0 \\
            -f^{14} & -f^{24} & u^2-f^{34} & 0 & -u^3+f^{45} & f^{46} \\
            -f^{15} & -f^{14} & u^1-f^{35} & u^3-f^{45} & 0 & f^{56} \\
            0 & 0 & 0 & -f^{46} & -f^{56} & 0 \\
        \end{pmatrix}.        
    \end{aligned}
\end{equation*}

Finally, from at the classification reported in~\cite{snowin}, we have 6 cases of solvable Lie algebras whose operator is non-degenerate:
\begin{itemize}
    \item $\mathfrak{s}_{6,k}$, with $k=162,\dots,167$, whose operator is non-degenerate:
\begin{equation*}
\begin{aligned}
 \mathcal{A}_{6,11}&=
\left(
\begin{array}{c@{\hspace{3pt}}c@{\hspace{3pt}}c@{\hspace{3pt}}c@{\hspace{3pt}}c@{\hspace{3pt}}c@{\hspace{3pt}}}
 0 & 0 & 0 & 0 & 0 & g_{16} \\
 0 & 0 & 0 & g_{16} & 0 & 0 \\
 0 & 0 & 0 & 0 & \frac{g_{16}}{a} & 0 \\
 0 & g_{16} & 0 & 0 & 0 & 0 \\
 0 & 0 & \frac{g_{16}}{a} & 0 & 0 & 0 \\
 g_{16} & 0 & 0 & 0 & 0 & g_{66} \\
\end{array}
\right)\partial_x \\&+ \left(
\begin{array}{c@{\hspace{3pt}}c@{\hspace{3pt}}c@{\hspace{3pt}}c@{\hspace{3pt}}c@{\hspace{3pt}}c@{\hspace{3pt}}}
 0 & 0 & 0 & 0 & 0 & 0 \\
 0 & 0 & 0 & -u^1+f^{24} & 0 & u^2+f^{26} \\
 0 & 0 & 0 & 0 & -u^1+f^{35} & a u^3+f^{36} \\
 0 & u^1-f^{24} & 0 & 0 & 0 & -u^4+f^{46} \\
 0 & 0 & u^1-f^{35} & 0 & 0 & -a u^5+f^{56} \\
 0 & -u^2-f^{26} & -a u^3-f^{36} & u^4-f^{46} & a u^5-f^{56} & 0 \\
\end{array}
\right),
\end{aligned}
\end{equation*}
\begin{equation*}
\begin{aligned}
 \mathcal{A}_{6,12}&=
\left(
\begin{array}{c@{\hspace{3pt}}c@{\hspace{3pt}}c@{\hspace{3pt}}c@{\hspace{3pt}}c@{\hspace{3pt}}c@{\hspace{3pt}}}
 0 & 0 & 0 & 0 & 0 & g_{16} \\
 0 & 0 & 0 & g_{16} & -g_{16} & 0 \\
 0 & 0 & 0 & 0 & g_{16} & 0 \\
 0 & g_{16} & 0 & 0 & 0 & 0 \\
 0 & -g_{16} & g_{16} & 0 & 0 & 0 \\
 g_{16} & 0 & 0 & 0 & 0 & g_{66} \\
\end{array}
\right)\partial_x \\&+ \left(
\begin{array}{c@{\hspace{3pt}}c@{\hspace{3pt}}c@{\hspace{3pt}}c@{\hspace{3pt}}c@{\hspace{3pt}}c@{\hspace{3pt}}}
 0 & 0 & 0 & 0 & 0 & 0 \\
 0 & 0 & 0 & -u^1+f^{24} & f^{25} & +u^2+u^3+f^{26} \\
 0 & 0 & 0 & 0 & -u^1+f^{24} & u^3+f^{36} \\
 0 & u^1-f^{24} & 0 & 0 & 0 & -u^4+f^{46} \\
 0 & -f^{25} & u^1-f^{24} & 0 & 0 & -u^4-u^5+f^{56} \\
 0 & -u^2-u^3-f^{26} & -u^3-f^{36} & u^4-f^{46} & +u^4+u^5-f^{56} & 0 \\
\end{array}
\right),
\end{aligned}
\end{equation*}
\begin{equation*}
\begin{aligned}
 \mathcal{A}_{6,13}&=
\left(
\begin{array}{c@{\hspace{3pt}}c@{\hspace{3pt}}c@{\hspace{3pt}}c@{\hspace{3pt}}c@{\hspace{3pt}}c@{\hspace{3pt}}}
 0 & 0 & 0 & 0 & 0 & g_{16} \\
 0 & 0 & 0 & \frac{g_{16}}{\alpha } & 0 & 0 \\
 0 & 0 & g_{16} & 0 & 0 & 0 \\
 0 & \frac{g_{16}}{\alpha } & 0 & 0 & 0 & 0 \\
 0 & 0 & 0 & 0 & g_{16} & 0 \\
 g_{16} & 0 & 0 & 0 & 0 & g_{66} \\
\end{array}
\right)\partial_x \\&+ \left(
\begin{array}{c@{\hspace{3pt}}c@{\hspace{3pt}}c@{\hspace{3pt}}c@{\hspace{3pt}}c@{\hspace{3pt}}c@{\hspace{3pt}}}
 0 & 0 & 0 & 0 & 0 & 0 \\
 0 & 0 & 0 & -u^1+f^{24} & 0 & \alpha  u^2+f^{26} \\
 0 & 0 & 0 & 0 & -u^1+f^{35} & u^5+f^{36} \\
 0 & u^1-f^{24} & 0 & 0 & 0 & -\alpha  u^4+f^{46} \\
 0 & 0 & u^1-f^{35} & 0 & 0 & -u^3+f^{56} \\
 0 & -\alpha  u^2-f^{26} & -u^5-f^{36} & \alpha  u^4-f^{46} & u^3-f^{56} & 0 \\
\end{array}
\right),
\end{aligned}
\end{equation*}
\begin{equation*}
\begin{aligned}
\footnotesize \mathcal{A}_{6,14}&=
\left(
\begin{array}{c@{\hspace{3pt}}c@{\hspace{3pt}}c@{\hspace{3pt}}c@{\hspace{3pt}}c@{\hspace{3pt}}c@{\hspace{3pt}}}
 0 & 0 & 0 & 0 & 0 & \left(\alpha ^2+1\right) (-g_{25}) \\
 0 & 0 & 0 & \alpha  (-g_{25}) & g_{25} & 0 \\
 0 & 0 & 0 & -g_{25} & \alpha  (-g_{25}) & 0 \\
 0 & \alpha  (-g_{25}) & -g_{25} & 0 & 0 & 0 \\
 0 & g_{25} & \alpha  (-g_{25}) & 0 & 0 & 0 \\
 \left(\alpha ^2+1\right) (-g_{25}) & 0 & 0 & 0 & 0 & g_{66} \\
\end{array}
\right)\partial_x \\&+ \left(
\begin{array}{c@{\hspace{3pt}}c@{\hspace{3pt}}c@{\hspace{3pt}}c@{\hspace{3pt}}c@{\hspace{3pt}}c@{\hspace{3pt}}}
 0 & 0 & 0 & 0 & 0 & 0 \\
 0 & 0 & 0 & -u^1 & 0 & \alpha  u^2+u^3 \\
 0 & 0 & 0 & 0 & -u^1 & \alpha u^3-u^2 \\
 0 & u^1 & 0 & 0 & 0 & -\alpha  u^4+u^5 \\
 0 & 0 & u^1 & 0 & 0 & -\alpha  u^5-u^4 \\
 0 & -\alpha  u^2-u^3& -\alpha  u^3+u^2 & \alpha  u^4-u^5 & \alpha  u^5+u^4 & 0 \\
\end{array}
\right)\\&+ \left(
\begin{array}{c@{\hspace{3pt}}c@{\hspace{3pt}}c@{\hspace{3pt}}c@{\hspace{3pt}}c@{\hspace{3pt}}c@{\hspace{3pt}}}
 0 & 0 & 0 & 0 & 0 & 0 \\
 0 & 0 & 0 & f^{24} & f^{25} & f^{26} \\
 0 & 0 & 0 & -f^{25} & f^{24} & f^{36} \\
 0 & -f^{24} & f^{25} & 0 & 0 & f^{46} \\
 0 & -f^{25} & -f^{24} & 0 & 0 & f^{56} \\
 0 & -f^{26} & -f^{36} & -f^{46} & -f^{56} & 0 \\
\end{array}
\right)
\end{aligned}
\end{equation*}
\begin{equation*}
\begin{aligned}
\scriptstyle \mathcal{A}_{6,15}&=
\left(
\begin{array}{c@{\hspace{3pt}}c@{\hspace{3pt}}c@{\hspace{3pt}}c@{\hspace{3pt}}c@{\hspace{3pt}}c@{\hspace{3pt}}}
 0 & 0 & 0 & 0 & 0 & g_{16} \\
 0 & g_{16} & 0 & 0 & 0 & 0 \\
 0 & 0 & \frac{g_{16}}{a} & 0 & 0 & 0 \\
 0 & 0 & 0 & g_{16} & 0 & 0 \\
 0 & 0 & 0 & 0 & \frac{g_{16}}{a} & 0 \\
 g_{16} & 0 & 0 & 0 & 0 & g_{66} \\
\end{array}
\right)\partial_x \\&+ \left(
\begin{array}{c@{\hspace{3pt}}c@{\hspace{3pt}}c@{\hspace{3pt}}c@{\hspace{3pt}}c@{\hspace{3pt}}c@{\hspace{3pt}}}
 0 & 0 & 0 & 0 & 0 & 0 \\
 0 & 0 & 0 & -u^1+f^{24} & 0 & u^4+f^{26} \\
 0 & 0 & 0 & 0 & -u^1+f^{35} & a u^5+f^{36} \\
 0 & u^1-f^{24} & 0 & 0 & 0 & -u^2+f^{46} \\
 0 & 0 & u^1-f^{35} & 0 & 0 & -a  u^3+f^{56} \\
 0 & -u^4-f^{26} & -a u^5-f^{36} & u^2-f^{46} & a  u^3-f^{56} & 0 \\
\end{array}
\right),
\end{aligned}
\end{equation*}
\begin{equation*}
\begin{aligned}
 \mathcal{A}_{6,16}&=
\left(
\begin{array}{c@{\hspace{3pt}}c@{\hspace{3pt}}c@{\hspace{3pt}}c@{\hspace{3pt}}c@{\hspace{3pt}}c@{\hspace{3pt}}}
 0 & 0 & 0 & 0 & 0 & g_{16} \\
 0 & -g_{16} & 0 & 0 & -g_{16} & 0 \\
 0 & 0 & -g_{16} & g_{16} & 0 & 0 \\
 0 & 0 & g_{16} & 0 & 0 & 0 \\
 0 & -g_{16} & 0 & 0 & 0 & 0 \\
 g_{16} & 0 & 0 & 0 & 0 & g_{66} \\
\end{array}
\right)\partial_x \\&+ \left(
\begin{array}{c@{\hspace{3pt}}c@{\hspace{3pt}}c@{\hspace{3pt}}c@{\hspace{3pt}}c@{\hspace{3pt}}c@{\hspace{3pt}}}
 0 & 0 & 0 & 0 & 0 & 0 \\
 0 & 0 & f^{23} & -u^1f^{24} & 0 & u^3+u^4+f^{26} \\
 0 & -f^{23} & 0 & 0 & -u^1+f^{24} & -u^2+u^5+f^{36} \\
 0 & u^1-f^{24} & 0 & 0 & 0 & u^5+f^{46} \\
 0 & 0 & u^1-f^{24} & 0 & 0 & -u^4+f^{56} \\
 0 & -u^3-u^4-f^{26} & u^2-u^5-f^{36} & -u^5-f^{46} & u^4-f^{56} & 0 \\
\end{array}
\right),
\end{aligned}
\end{equation*}
with $0<|a|\le 1$, $\alpha>0$.
\end{itemize}

Finally, the last two cases are represented by the simple Lie algebra 
$\mathfrak{so}(1,3,\mathbb{R})$ whose associated operator is: 
\begin{align*}
    \mathcal{A}_{6,17} &=
    \begin{pmatrix}
        -g_{55} & 0 & 0 & g_{25} & 0 & 0 \\
        0 & -g_{55} & 0 & 0 & g_{25} & 0 \\
        0 & 0 & -g_{55} & 0 & 0 & g_{25} \\
        g_{25} & 0 & 0 & g_{55} & 0 & 0 \\
        0 & g_{25} & 0 & 0 & g_{55} & 0 \\
        0 & 0 & g_{25} & 0 & 0 & g_{55}
    \end{pmatrix}
    \partial_x \notag
    \\
    &+ 
    \begin{pmatrix}
        0 & -u^3 & u^2 & 0 & -u^6 & u^5 \\
        u^3 & 0 & -u^1 & u^6 & 0 & -u^4 \\
        -u^2 & u^1 & 0 & -u^5 & u^4 & 0 \\
        0 & -u^6 & u^5 & 0 & u^3 & -u^2 \\
        u^6 & 0 & -u^4 & -u^3 & 0 & u^1 \\
        -u^5 & u^4 & 0 & u^2 & -u^1 & 0
    \end{pmatrix}
    \\
    &+ 
    \begin{pmatrix}
        0 & f^{12} & f^{13} & 0 & f^{15} & f^{16} \\
        -f^{12} & 0 & f^{23} & -f^{15} & 0 & f^{26} \\
        -f^{13} & -f^{23} & 0 & -f^{16} & -f^{26} & 0 \\
        0 & +f^{15} & f^{16} & 0 & -f^{12} & -f^{13} \\
        -f^{15} & 0 & f^{26} & f^{12} & 0 & -f^{23} \\
        -f^{16} & -f^{26} & 0 & f^{13} & f^{23} & 0 \\
    \end{pmatrix},\notag
    \end{align*}
and the Levi decomposable Lie algebra 
$\mathfrak{sl}(2,\mathbb{R})\ltimes 3\mathfrak{n}_{1,1}$ with associated
operator:
\begin{align*}
\mathcal{A}_{6,18} &=
\begin{pmatrix}
 g_{22} & 0 & 0 & g_{25} & 0 & 0 \\
 0 & g_{22} & 0 & 0 & g_{25} & 0 \\
 0 & 0 & g_{22} & 0 & 0 & g_{25} \\
 g_{25} & 0 & 0 & 0 & 0 & 0 \\
 0 & g_{25} & 0 & 0 & 0 & 0 \\
 0 & 0 & g_{25} & 0 & 0 & 0 \\
\end{pmatrix}
\partial_x\notag 
\\
&+
\begin{pmatrix}
 0 & -u^3 & u^2 & 0 & -u^6 & u^5 \\
 u^3 & 0 & -u^1 & u^6 & 0 & -u^4 \\
 -u^2 & u^1 & 0 & -u^5 & u^4 & 0 \\
 0 & -u^6 & u^5 & 0 & 0 & 0 \\
 u^6 & 0 & -u^4 & 0 & 0 & 0 \\
 -u^5 & u^4 & 0 & 0 & 0 & 0 \\
\end{pmatrix}
\\
&+
\begin{pmatrix}
 0 & f^{12} & f^{13} & 0 & f^{15} & f^{16} \\
 -f^{12} & 0 & f^{23} &-f^{15} & 0 & f^{26} \\
 -f^{13} & -f^{23} & 0 & -f^{16} & -f^{26} & 0 \\
 0 & f^{15} & f^{16} & 0 & 0 & 0 \\
 -f^{15} & 0 & f^{26} & 0 & 0 & 0 \\
 -f^{16} & -f^{26} & 0 & 0 & 0 & 0 \\
\end{pmatrix}.\notag
\end{align*}

\newpage 

 \providecommand{\cprime}{\/{\mathsurround=0pt$'$}}
  \providecommand*{\SortNoop}[1]{}

\end{document}